\documentclass[a4paper,11pt,english]{amsart}
\usepackage[utf8]{inputenc}             
\usepackage[left = 25mm,top = 30mm,bottom = 25mm,right = 25mm]{geometry}
\usepackage{amsmath,amssymb,amsfonts,dsfont,amsthm}
\usepackage{lmodern}
\usepackage{cmap}                       
\usepackage{babel}
\usepackage[pdfdisplaydoctitle = true,
      colorlinks = true,
      urlcolor = blue,
      citecolor = blue,
      linkcolor = blue,
      pdfstartview = FitH,
      pdfpagemode = UseNone,
      bookmarksnumbered = true]{hyperref} 
\usepackage{hyperxmp}
\usepackage{float}
\usepackage{booktabs}

\newtheorem{thm}{Theorem}[section]

\newtheorem{lem}[thm]{Lemma}

\theoremstyle{definition}
\newtheorem{defn}[thm]{Definition}
\theoremstyle{remark}
\newtheorem{rem}[thm]{Remark}
\newtheorem{ex}[thm]{Example}
\numberwithin{equation}{section}

\renewcommand{\vec}{\pmb}

\newcommand{\RR}{\mathbb{R}}                            
\newcommand{\NN}{\mathbb{N}}                            

\newcommand{\VV}{\mathbb{V}}                            
\newcommand{\Ela}{\mathbb{E}\mathrm{la}}                
\newcommand{\Piez}{\mathbb{P}\mathrm{iez}}              
\newcommand{\Magn}{\mathbb{M}\mathrm{agn}}              
\newcommand{\HH}{\mathbb{H}}                            
\newcommand{\TT}{\mathbb{T}}                            
\newcommand{\Sym}{\mathbb{S}}                           
\newcommand{\Pn}[1]{\mathcal{P}_{#1}}                   
\newcommand{\Hn}[1]{\mathcal{H}^{#1}}                   

\newcommand{\OO}{\mathrm{O}}                            
\newcommand{\SO}{\mathrm{SO}}                           
\newcommand{\Id}{\mathrm{I}}                            

\newcommand{\ee}{\pmb{e}}                               
\newcommand{\nn}{\pmb{n}}                               
\newcommand{\ww}{\pmb{w}}                               
\newcommand{\xx}{\pmb{x}}                               

\newcommand{\bnu}{\pmb{\nu}}                            
\newcommand{\btau}{\pmb{\tau}}                          

\newcommand{\bmu}{\pmb{\mu}}

\newcommand{\bepsilon}{\pmb{\epsilon}}
\newcommand{\bsigma}{\pmb{\sigma}}
\newcommand{\bomega}{\pmb{\omega}}

\newcommand{\bPi}{\pmb{\Pi}}

\newcommand{\hrho}{\hat{\rho}}

\newcommand{\rp}{\mathrm{p}}                            
\newcommand{\rh}{\mathrm{h}}                            

\newcommand{\ba}{\mathbf{a}}
\newcommand{\bb}{\mathbf{b}}

\newcommand{\bd}{\mathbf{d}}

\newcommand{\br}{\mathbf{r}}
\newcommand{\bs}{\mathbf{s}}

\newcommand{\bv}{\mathbf{v}}

\newcommand{\by}{\mathbf{y}}
\newcommand{\bz}{\mathbf{z}}

\newcommand{\lc}{\pmb\varepsilon}                       
\newcommand{\bC}{\mathbf{C}}                            
\newcommand{\bH}{\mathbf{H}}                            
\newcommand{\bK}{\mathbf{K}}                            %
\newcommand{\bL}{\mathbf{L}}                            %
\newcommand{\bS}{\mathbf{S}}                            
\newcommand{\bT}{\mathbf{T}}                            
\newcommand{\bF}{\mathbf{F}}                            

\newcommand{\bP}{\mathbf{P}}

\DeclareMathOperator{\tr}{tr}
\DeclareMathOperator{\2dots}{:}

\newcommand{\norm}[1]{\left\Vert#1\right\Vert}  
\newcommand{\set}[1]{\left\{#1\right\}}         

\newcommand{\otimesbar}{\; \underline{\overline{\otimes}} \;}
\newcommand{\rcont}{\overset{(r)}{\cdot}} 
\newcommand{\poisson}[2]{\left\{ #1 , #2 \right\}} 
\newcommand{\contsnu}[2]{\bnu^{#1} \odot\left(\bS\overset{(#2)}{\cdot}\bnu^{#2}\right)}

\hypersetup{
  pdfauthor = {M. Olive, B. Desmorat, B. Kolev, R. Desmorat}, 
  pdftitle = {Reduced algebraic conditions for plane/axial tensorial symmetries}, 
  pdfsubject = {}, 
  pdfkeywords = {Anisotropy; Elasticity; Piezo-electricity; Piezo-magnetism} 
  pdflang = en, 
  }
\begin{document}

\title{Reduced algebraic conditions for plane/axial tensorial symmetries}

\author{M. Olive}
\address[Marc Olive]{Universit\'{e} Paris-Saclay, ENS Paris-Saclay, CNRS, LMT - Laboratoire de M\'{e}canique et Technologie, 94235, Cachan, France}
\email{marc.olive@math.cnrs.fr}

\author{B. Desmorat}
\address[Boris Desmorat]{Sorbonne Universit\'{e}, UMPC Univ Paris 06, CNRS, UMR 7190, Institut d'Alembert, F-75252 Paris Cedex 05, France \& Univ Paris Sud 11, F-91405 Orsay, France}
\email{boris.desmorat@sorbonne-universite.fr}

\author{B. Kolev}
\address[Boris Kolev]{Universit\'{e} Paris-Saclay, ENS Paris-Saclay, CNRS, LMT - Laboratoire de M\'{e}canique et Technologie, 94235, Cachan, France}
\email{boris.kolev@math.cnrs.fr}

\author{R. Desmorat}
\address[Rodrigue Desmorat]{Universit\'{e} Paris-Saclay, ENS Paris-Saclay, CNRS, LMT - Laboratoire de M\'{e}canique et Technologie, 94235, Cachan, France}
\email{rodrigue.desmorat@ens-paris-saclay.fr}

\begin{abstract}
  In this article, we formulate necessary and sufficient \emph{polynomial equations} for the existence of a symmetry plane or an order-two axial symmetry for a totally symmetric tensor of order $n \ge 1$. These conditions are effective and of degree $n$ (the tensor's order) in the components of the normal to the plane (or the direction of the axial symmetry). These results are then extended to obtain necessary and sufficient polynomial conditions for the existence of such symmetries for an Elasticity tensor, a Piezo-electricity tensor or a Piezo-magnetism pseudo-tensor.
\end{abstract}


\thanks{R. Desmorat, B. Kolev and M.Olive were partially supported by CNRS Projet 80|Prime GAMM (G\'{e}om\'{e}trie alg\'{e}brique complexe/r\'{e}elle et m\'{e}canique des mat\'{e}riaux)}%

\date{April 20, 2020}

\maketitle

\section{Introduction}
\label{sec:introduction}

In classical 3D Linear Elasticity, plane symmetries are of primary importance either to characterize symmetry classes~\cite{CVC2001} or to determine the propagation axes of longitudinal waves~\cite{Nor1989,GG1999}. Underlying these mechanical properties is the \emph{linear representation} $\rho$ \cite{GSS1988,FH2013} of the orthogonal group $\OO(3)$ on the space of Elasticity tensors $\Ela$~\cite{FV1996}, and we write
\begin{equation*}
  \overline{\bC} = \rho(g)\bC, \qquad g \in \OO(3), \quad \bC \in \Ela .
\end{equation*}
A symmetry plane of an Elasticity tensor $\bC$ corresponds to the symmetry $g = \bs(\bnu)$ with respect to the plane $\bnu^\perp$, where
\begin{equation*}
  \bs(\bnu) : = \Id-2 \bnu \otimes \bnu, \qquad \norm{\bnu} = 1, \qquad \det \bs(\bnu) = -1 .
\end{equation*}
If $\bnu = (x,y,z)\in \RR^{3}$ in some basis of $\RR^{3}$, the equation $\rho(\bs(\bnu))\bC = \bC$, recasts into $21$ homogeneous polynomial equations of degree $8$ in the components $(x,y,z)$ of $\bnu$. The number of solutions of these equations, as well as their respective angles, provide direct information on the \emph{symmetry group} of the given tensor, and therefore of its symmetry class. Besides, each solution gives rise to a direction of propagation of a longitudinal wave.

Therefore, two natural questions arise from these considerations, which are still meaningful for higher order tensors. The first one concerns the explicit determination of the symmetry group of a given tensor (and thus of its symmetry class). The second one concerns the determination of the directions of longitudinal wave propagation.

To determine the symmetry group of an elasticity tensor, some authors have used Kelvin's representation, its spectral decomposition and its \emph{eigenstrains}. Unfortunately, this approach introduces many cases related to multiplicity of eigenvalues and does not give necessary and sufficient conditions for the characterisation of its symmetry class. Another approach has been considered by Fran\c{c}ois and al~\cite{FGB1998}. They used a distance function to characterize symmetry planes and produced pole figures to characterize each class. However, extending such an approach to $n$th-order tensors requires solving equations of degree $2n$ in $\bnu$.

Many studies have focused on the calculation of the directions of wave propagation, either in Elasticity~\cite{Kha1962,Fed1968,BH1998} or in the case of coupled phenomena, such as Piezo-electricity or Piezo-magnetism~\cite{Aul1981,BJC2008,PLWF2008}. Indeed, any symmetry plane $\bnu^{\perp}$ of an Elasticity tensor $\bC$ gives rise to a propagation direction $\bnu$ of a longitudinal wave~\cite{Sti1965,Kol1966,Fed1968,Nor1989}, also called an \emph{acoustic axis}.

Among these studies stand out the famous Cowin--Mehrabadi conditions~\cite{CM1987,Cow1989} for the existence of a symmetry plane of an Elasticity tensor, \textit{i.e.} a fourth-order tensor $\bC$, with index symmetries
\begin{equation*}
  C_{ijkl} = C_{jikl} = C_{klij}.
\end{equation*}
These conditions were first expressed by introducing the two independent traces of $\bC$, the \emph{dilatation} tensor $\bd$ and \emph{Voigt's} tensor $\bv$:
\begin{equation*}
  \bd : = \tr_{12} \bC \quad (d_{ij} = C_{ppij}), \qquad \bv : = \tr_{13}\bC \quad (v_{ij} = C_{pipj}),
\end{equation*}
which are symmetric second-order tensors.

\begin{thm}[Cowin--Mehrabadi (1987)]\label{thm:CM1987}
  Let $\bnu$ be a unit vector. Then $\bnu^{\perp}$ is a symmetry plane of an Elasticity tensor $\bC\in \Ela$ if and only if
  \begin{align}
    \left[ \left(\bnu\cdot \bC\cdot \bnu\right) \bnu\right] \times \bnu  & = 0 \label{eq:CM1987-1}
    \\
    \left[ \left(\btau\cdot \bC\cdot \btau\right)\bnu\right] \times \bnu & = 0 \label{eq:CM1987-2}
    \\
    \left( \bd\, \bnu\right) \times \bnu                                 & = 0 \label{eq:CM1987-3}
    \\
    \left( \bv \, \bnu\right)\times \bnu                                 & = 0 \label{eq:CM1987-4}
  \end{align}
  for all unit vectors $\btau$ perpendicular to $\bnu$.
\end{thm}

It appeared later that the third and fourth conditions in theorem~\ref{thm:CM1987} are in fact consequences of the first two ones~\cite{Cow1989,Nor1989}. Nevertheless, when non trivially satisfied, conditions \eqref{eq:CM1987-3} and \eqref{eq:CM1987-4} give candidates --- the common eigenvectors of $\bd$ and $\bv$ --- to be normals defining the symmetry planes of $\bC$. On the other side, the fact that~\eqref{eq:CM1987-2} needs to be checked for an infinity of $\btau \perp \bnu$, makes this criteria not really constructive. For instance, for a cubic tensor $\bC\in \Ela$, both $\bd$ and $\bv$ are spherical, so that \eqref{eq:CM1987-3} and \eqref{eq:CM1987-4} are trivially satisfied and theorem~\ref{thm:CM1987} is not very useful, in that case. Such a drawback is also present in the equivalent forms used by~\cite{Nor1989,Jar1994} and in the generalized Cowin-Mehrabadi theorems~\cite{Tin2003}.

Note also the following fact. The direct necessary and sufficient condition
\begin{equation*}
  \rho(\bs(\bnu))\bC = \bC
\end{equation*}
for the existence of a symmetry plane $\bnu^\perp$ are polynomial equations of degree 8 in the components of $\bnu$, whereas Cowin--Mehrabadi's conditions are of degree at most 4 in $\bnu$ and for these reasons could be called \emph{reduced equations}.

In the present work, we propose to overcome these difficulties by formulating \emph{new reduced algebraic equations} which are necessary and sufficient conditions for a unit vector $\bnu$ to define a symmetry plane of a tensor of any order. This new approach is also constructive, and does not involve vectors $\btau$ orthogonal to $\bnu$.

To illustrate our approach, consider a symmetric second-order tensor $\ba$ (order $n = 2$). In this case, the direct approach --- considering $\ba$ unchanged by the action of $\bs(\bnu)$, \emph{i.e.} solving $\rho(\bs(\bnu))\ba = \ba$ (see~equation~\eqref{eq:rhosaegala_pol} in section \ref{sec:reduced-equations}) --- leads to 6 algebraic equations of degree 4 ($ = 2n$) in $\bnu$. But one can notice that $\bnu^\perp$ is a symmetry plane if and only if $\bnu$ is an eigenvector of $\ba$, which writes as
\begin{equation*}
  (\ba\,\bnu)\times \bnu = 0,
\end{equation*}
and leads thus to 3 algebraic equations of degree 2 (the order of the tensor). In short, the present work is a generalization of this simple observation.

We shall first solve the problem for a totally symmetric tensor or pseudo-tensor of any order $n$, and then for any tensor using the harmonic decomposition. Our main result, theorem~\ref{thm:main} states necessary and sufficient conditions for the existence of a symmetry plane or an order-two axial symmetry for a given totally symmetric tensor or pseudo-tensor of order $n$. Moreover, these conditions are algebraic equations of degree $n$ in $\bnu$, instead of degree $2n$ in the direct approach. Our work is based on a strong link between \emph{harmonic tensors} on $\RR^{3}$ and three-variable \emph{harmonic polynomials}~\cite{Bae1993}, \textit{i.e.} homogeneous three-variable polynomials with vanishing Laplacian. It involves some \emph{covariant} operations, such as the \emph{generalized cross product} as defined in~\cite{OKDD2018a}, and allows us to obtain particularly condensed equations, as illustrated by the condition $(\ba\,\bnu)\times \bnu = 0$ for second-order symmetric tensors.

To illustrate the strength and the generality of this approach, we finally apply it to several important mechanical situations: to \emph{Elasticity} tensors (\autoref{sec:elasticity-tensors}), to \emph{Piezo-electricity} tensors (\autoref{sec:Piezo-electricity-tensors}) and to \emph{Piezo-magnetism} pseudo-tensors (\autoref{sec:Piezo-magnetism-tensors}). Some applications are also provided to \emph{fabric tensors} (\autoref{sec:fabric-tensors}), which are used in continuum mechanics for the descriptions of crack density~\cite{Ona1984}, of anisotropic contacts and grains orientations within granular materials~\cite{Oda1982}, of the anisotropy of biological tissues and bones~\cite{JT1984} and of the tensorial representation of the rafting phenomenon in single crystal superalloys at high temperature~\cite{CDC2018}.

The outline of the paper is the following. In~\autoref{sec:notations-and-backgrounds}, we fix notations and recall mathematical/geometrical backgrounds. Our main result, theorem~\ref{thm:main} is stated and proved in~\autoref{sec:reduced-equations}. Finally, applications are provided to \emph{fabric tensors} in~\autoref{sec:fabric-tensors}, \emph{Elasticity} in~\autoref{sec:elasticity-tensors}, \emph{Piezo-electricity} in~\autoref{sec:Piezo-electricity-tensors} and \emph{Piezo-magnetism} in~\autoref{sec:Piezo-magnetism-tensors}.

\section{Notations and geometrical backgrounds}
\label{sec:notations-and-backgrounds}

We define $\TT^{n}(\RR^{3}) : = \otimes^{n}(\RR^{3})$ as the vector space of $n$th-order tensors of the Euclidean space $\RR^{3}$, $\Sym^{n}(\RR^{3})$ as the subspace of $\TT^{n}(\RR^{3})$ of totally symmetric tensors of order $n$ and $\Lambda^{n}(\RR^{3})$ as the subspace of $\TT^{n}(\RR^{3})$ of alternate $n$th-order tensors. Note that $\Lambda^{n}(\RR^{3}) = \set{0}$, if $n > 3$, and that $\Lambda^{3}(\RR^{3})$ is one-dimensional.  The third-order \emph{Levi-Civita tensor} $\lc$ is a basis of $\Lambda^{3}(\RR^{3})$ and, in an orthonormal basis $(\ee_i)$, its components are
\begin{equation*}
  \lc=(\varepsilon_{ijk}), \qquad \varepsilon_{ijk}=\det(\ee_i, \ee_j, \ee_k).
\end{equation*}

\subsection{Natural and twisted tensorial representations}
\label{subsec:tensorial-representations}

The canonical action of the orthogonal group $\OO(3)$ on $\RR^{3}$ induces a linear action (also called a linear representation in mathematics) $\rho_n$ on the vector space $\TT^{n}(\RR^{3})$ of $n$th-order tensors. This representation is defined as follows
\begin{equation*}
  (\rho_{n}(g)\bT)(\xx_{1},\dotsc , \xx_{n}) : = \bT(g^{-1}\xx_{1},\dotsc , g^{-1}\xx_{n}), \qquad g \in \OO(3),\, \bT \in \TT^{n}(\RR^{3}),
\end{equation*}
or in components as
\begin{equation*}
  (\rho_{n}(g)\bT)_{i_{1} i_{2}\dotsb i_{n}} = g_{i_{1}}^{j_{1}}g_{i_{2}}^{j_{2}} \dotsb g_{i_{n}}^{j_{n}} \, T_{j_{1} j_{2}\cdots j_{n}}.
\end{equation*}
It is called the \emph{natural tensorial representation} of $\OO(3)$ on $\TT^{n}(\RR^{3})$.

If $\VV$ is a subspace of $\TT^{n}(\RR^{3})$ which is \emph{stable} under $\OO(3)$, which means that
\begin{equation*}
  \rho_{n}(g)(\VV) \subset \VV, \qquad \forall g \in \OO(3),
\end{equation*}
the restriction of $\rho_{n}$ to $\VV$ induces a representation of $\OO(3)$ on $\VV$ that we will still denote by $\rho_{n}$. This applies in particular to $\VV = \Sym^{n}(\RR^{3})$ and $\VV = \Lambda^{n}(\RR^{3})$.

\begin{ex}
  For instance, $\rho_{3}$ induces a representation of $\OO(3)$ on $\Lambda^{3}(\RR^{3})$, the one-dimensional subspace of $\TT^{3}(\RR^{3})$ of alternate third-order tensors. This induced representation is given by
  \begin{equation*}
    \rho_{3}(g) \bmu = (\det g)\, \bmu, \qquad \bmu \in \Lambda^{3}(\RR^{3}).
  \end{equation*}
\end{ex}

Given a linear representation $\rho$ of $\OO(3)$ on $\VV$, one can build a new one on $\Lambda^{3}(\RR^{3}) \bigotimes \VV$, defined as
\begin{equation*}
  (\rho_{3} \otimes \rho)(g) \, \bmu \otimes \bT : = (\rho_{3}(g)\bmu) \otimes (\rho(g)\bT) = (\det g)\, \bmu \otimes (\rho(g)\bT), \qquad \bmu \in \Lambda^{3}(\RR^{3}), \,\bT \in \VV.
\end{equation*}
Moreover, $\Lambda^{3}(\RR^{3}) \bigotimes \VV$ and $\VV$ have the same dimension and are isomorphic (after the choice of an oriented euclidean structure on $\RR^{3}$). Hence, we may think of this new representation, that we will denote by $\hrho$, as a representation of $\OO(3)$ on $\VV$ itself:
\begin{equation}
  \hat{\rho}(g)\bT : = (\det g)\, \rho(g) \bT, \qquad \bT \in \VV.
\end{equation}
The representation $\hat{\rho}$ will be called the \emph{twisted tensorial representation}.

Note that when restricted to $\SO(3)$, the two representations $\rho$ and $\hrho$ are the same:
\begin{equation*}
  \hat{\rho}(g) \bT = \rho(g)\bT \qquad \forall \, g\in \SO(3),\, \forall \bT \in \VV.
\end{equation*}

\begin{rem}\label{rem:moinsId}
  The distinction between the tensorial representations $\rho_{n}$ and $\hrho_{n}$ is illustrated by the action of the central symmetry $-I \in \OO(3)$, where $I$ the identity, as we have
  \begin{equation*}
    \rho(-I) \bT = (-1)^{n} \bT,\qquad \hrho(-I) \bT = (-1)^{n+1} \bT,\quad \bT\in \TT^{n}(\RR^{3}).
  \end{equation*}
  In particular:
  \begin{itemize}
    \item if $\bT \in \TT^{2p}$, $\rho_{2p}(-I) \bT = \bT$ and the natural $\OO(3)$ representation reduces to the $\SO(3)$ representation on even-order tensors;
    \item if $\bT \in \TT^{2p+1}$, $\hrho_{2p+1}(-I) \bT = \bT$ and the $\OO(3)$ twisted representation reduces to the $\SO(3)$ representation on odd-order pseudo-tensors.
  \end{itemize}
\end{rem}

\begin{ex}
  There is a well-known $\OO(3)$-equivariant isomorphism between the natural representation $\rho_{2}$ on the space of skew-symmetric second-order tensors $\Lambda^{2}(\RR^{3})$ and the twisted representation $\hrho_{1}$ on $\RR^{3}$. It is given by
  \begin{equation*}
    \bomega \mapsto \lc \2dots \bomega, \qquad \Lambda^{2}(\RR^{3}) \to \RR^{3}.
  \end{equation*}
  where $\lc = (\epsilon_{ijk})$ is the {Levi-Civita tensor} (\textit{i.e} the determinant in the canonical basis) and the contraction $\lc:\bomega$ is defined by
  \begin{equation*}
    (\lc: \bomega)_{i} = \varepsilon_{ijk}\omega_{jk}.
  \end{equation*}
\end{ex}

\begin{rem}
  It is necessary to construct mathematical objects that allow to translate the fact that some physical phenomenon indeed depend on the orientation of space. It is therefore natural to consider tensors that are changed in a non--standard way when the orientation of space changes. These tensors are usually called \emph{pseudo--tensors} or \emph{axial--tensors}~\cite{BT1968}. They correspond to the twisted representations $\hrho_n$ on $\TT^{n}(\RR^{3})$. Examples of such objects are the magnetic field $\vec H$, the magnetization $\vec M$ which correspond to the representation $\hrho_{1}$ on $\RR^{3}$ and the Piezo-magnetism pseudo-tensor $\bPi$ which corresponds to the representation $\hrho_{3}$ on $\TT^{3}(\RR^{3})$ (see Section~\ref{sec:Piezo-magnetism-tensors}).
\end{rem}

In the following, when we do not want to specify which tensorial representation is considered we will use the generic notation $\varrho$ to design either $\rho$ or $\hrho$.

\subsection{Totally symmetric tensors and homogeneous polynomials}
\label{subsec:tot-sym-and-hom-pol}

There is a well-known isomorphism between the space $\Sym^{n}(\RR^{3})$ of totally symmetric tensors and the space $\Pn{n}(\RR^{3})$ of $n$th degree \emph{homogeneous polynomials}
\begin{equation*}
  \varphi : \Sym^{n}(\RR^{3}) \to \Pn{n}(\RR^{3}), \qquad \bS \mapsto \rp(\xx) : = \bS(\xx, \dotsc , \xx).
\end{equation*}
In any basis $(\ee_{i})$, we have
\begin{equation*}
  \varphi(\bS)(\xx) = S_{i_{1}i_{2}\dotsc i_{n}}x_{i_{1}}x_{i_{2}} \dotsb x_{i_{n}}, \qquad \xx : = (x_{1},x_{2},x_3)\in \RR^{3}.
\end{equation*}
The inverse operation $\bS = \varphi^{-1}(\rp)\,\in \Sym^{n}(\RR^{3})$ is obtained explicitly by polarization (see~\cite{Gor2017,OKDD2018a}, or~\cite{Bae1993}).

\begin{ex}
  If $\ba = (a_{ij})$ is a symmetric second-order tensor, we get
  \begin{equation*}
    \varphi(\ba)(\xx) = a_{11}x_{1}^{2}+a_{22}x_{2}^{2}+a_{33}x_3^{2}+2a_{12}x_{1}x_{2}+2a_{13}x_{1}x_3+2a_{23}x_{2}x_3
  \end{equation*}
\end{ex}

\begin{rem}
  If one introduces the following actions of $\OO(3)$ on $\Pn{n}(\RR^{3})$:
  \begin{equation*}
    (\rho(g) \rp)(\xx) : = \rp(g^{-1}\xx), \qquad (\hrho(g)  \rp)(\xx) : = \det(g)\rp(g^{-1}\xx),
  \end{equation*}
  where $\rp \in \Pn{n}(\RR^{3})$, then, $\varphi$ becomes an equivariant isomorphism with respect to both actions. In other words:
  \begin{equation*}
    \varphi(\rho(g) \bS) = \rho(g) \varphi(\bS), \qquad \varphi (\hrho(g) \bS) = \hrho(g)  \varphi(\bS),
  \end{equation*}
  for all $g\in \OO(3)$ and $\bS\in \Sym^{n}(\RR^{3})$.
\end{rem}

\subsection{Harmonic decomposition}
\label{subsec:harmonic-decomposition}

\begin{defn}
  An $n$-th order totally symmetric and \emph{traceless} tensor will be called an \emph{harmonic tensor} and the subspace of $\Sym^{n}(\RR^{3})$ of harmonic tensors will be denoted by $\HH^{n}(\RR^{3})$ (or simply $\HH^{n}$, if there is no ambiguity).
\end{defn}

\begin{rem}
  In the correspondence between totally symmetric tensors and homogeneous polynomials, a \emph{traceless} totally symmetric tensor $\bH$ corresponds to an \emph{harmonic polynomial} $\rh$ (\textit{i.e.} with vanishing Laplacian: $\triangle \rh = 0$) and this justifies the appellation of \emph{harmonic tensor}. The space of homogeneous harmonic polynomials of degree $n$ will be denoted by $\Hn{n}(\RR^{3})$.
\end{rem}

The subspace $\HH^{n}(\RR^{3})$ of $\TT^{n}(\RR^{3})$ is stable under both linear representations $\rho_{n}$ and $\hrho_{n}$ and moreover \emph{irreducible} (its only invariant subspaces are itself and the null space). In the case of a symmetric second-order tensor, this decomposition corresponds to the usual decomposition into a deviator and a spherical tensor and we have an equivariant isomorphism between $\Sym^{2}(\RR^{3})$ and $\HH^{2}(\RR^{3}) \oplus \HH^{0}(\RR^{3})$, where $\HH^{2}(\RR^{3})$ corresponds to the deviatoric part and $\HH^{0}(\RR^{3})$, to the spherical part (of scalar component the trace of considered tensor). More precisely we have the following result (see also~\cite{Sch1951,Spe1970,JCB1978}).

\begin{thm}[Harmonic decomposition]
  Every finite dimensional representation $\VV$ of the orthogonal group $\OO(3)$ can be decomposed into a direct sum of irreducible representations, each of them being isomorphic to either $\rho_{n}$ or $\hrho_{n}$ on $\HH^{n}(\RR^{3})$, by an equivariant isomorphism.
\end{thm}

\subsection{Covariant operations on tensors}
\label{subsec:covariant-operations}

Given $\bT \in \TT^{n}(\RR^{3})$, we define the total symmetrization of $\bT$, noted $\bT^{s} \in \Sym^{n}(\RR^{3})$, as
\begin{equation*}
  \bT^{s}(\xx_{1},\dotsc,\xx_{n}) : = \frac{1}{n!}\sum_{\sigma \in \mathfrak{S}_{n}} \bT(\xx_{\sigma(1)},\dotsc,\xx_{\sigma(n)}).
\end{equation*}
This operation is \emph{covariant} for both the natural and the twisted representations, which means that
\begin{equation*}
  (\varrho(g)\bT)^{s} = \varrho(g)(\bT^{s}),\quad \forall g\in \OO(3),\quad \varrho = \rho\text{ or } \hrho.
\end{equation*}

The \emph{symmetric tensor product} between two symmetric tensors $\bS^{1} \in \Sym^{p}(\RR^{3})$ and $\bS^{2} \in \Sym^{q}(\RR^{3})$ is
\begin{equation*}
  \bS^{1} \odot \bS^{2} : = (\bS^{1} \otimes \bS^{2})^{s} \in \Sym^{p + q}(\RR^{3}).
\end{equation*}
This operation is \emph{covariant} for the natural representation, but not for the twisted one. For any $g\in \OO(3)$, we get indeed:
\begin{equation*}
  (\rho(g) \bS^{1}) \odot (\rho(g) \bS^{2}) = \rho(g) (\bS^{1} \odot \bS^{2}),\qquad
  (\hrho(g)\bS^{1}) \odot (\hrho(g) \bS^{2}) = \rho(g) (\bS^{1} \odot \bS^{2}).
\end{equation*}

The \emph{$r$-contraction} between totally symmetric tensors $\bS^{1} \in \Sym^{p}(\RR^{3})$ and $\bS^{2} \in \Sym^{q}(\RR^{3})$ is defined in any orthonormal basis as
\begin{equation*}
  (\bS^{1} \rcont \bS^{2})_{i_{1} \dotsb i_{p-r}j_{r+1} \dotsb j_{q}} = S^{1}_{i_{1}\dotsb i_{p-r}k_{1}\dotsb k_{r}} \, S^{2}_{k_{1}\dotsb k_{r}j_{r+1} \dotsb j_{q}},
\end{equation*}
if $r\leq \min(p,q)$ and is zero otherwise, where the summation convention on repeated indices is used.

This operation is \emph{covariant} for the natural representation, but not for the twisted one. For any $g\in \OO(3)$, we get indeed:
\begin{equation*}
  (\rho(g) \bS^{1}) \rcont (\rho(g) \bS^{2}) = \rho(g) (\bS^{1} \rcont \bS^{2}),\qquad
  (\hrho(g) \bS^{1}) \rcont (\hrho(g) \bS^{2}) = \rho(g) (\bS^{1} \rcont \bS^{2}).
\end{equation*}

For a given vector $\ww\in \RR^{3}$, we set
\begin{equation*}
  \ww^k : = \underbrace{\ww \otimes \ww \otimes \dotsb \otimes \ww}_\textrm{$k$ times},
\end{equation*}
with components $\ww^k_{i_{1} i_{2}\dotsc i_k} = w_{i_{1}}w_{i_{2}}\dotsc w_{i_k}$, so that the $r$-contraction between a $n$th-order symmetric tensor $\bS \in \Sym^{n}(\RR^{3})$ and $\ww^r$ (where $r\leq n$) reads
\begin{equation*}
  (\bS \rcont \ww^r)_{i_{1}i_{2}\dotsc i_{n-r}} = S_{i_{1}i_{2}\dotsc i_{n-r}j_{1}\dotsc j_{r}}w_{j_{1}}\dotsc w_{j_r},
  \qquad \bS\rcont \ww^r\in  \Sym^{n-r}.
\end{equation*}
The tensorial operation $(\bS,\ww)\mapsto \bS\rcont \ww^r$ ($r$-contraction with $\ww^r = \ww \otimes \dotsc \otimes \ww$) is $\SO(3)$-covariant but not always $\OO(3)$-covariant since
\begin{equation*}
  (\hrho(g) \ww)^r = (\det(g))^r \rho(g) (\ww^{r})
\end{equation*}
and thus
\begin{equation*}
  (\hrho(g)\bS) \rcont (\hrho(g) \ww)^r = (\det(g))^{r+1} \rho(g)(\bS\rcont \ww^{r}).
\end{equation*}

The \emph{generalized cross product} between two totally symmetric tensors $\bS^{1} \in \Sym^{p}(\RR^{3})$ and $\bS^{2} \in \Sym^{q}(\RR^{3})$ is defined as~\cite{OKDD2018}
\begin{equation}\label{eq:cross-product}
  \bS^{1} \times \bS^{2} : = - \left(\bS^{1}\cdot\lc \cdot \bS^{2}\right)^{s} \in \Sym^{p + q -1}(\RR^{3}).
\end{equation}
In any orthonormal basis, it writes as
\begin{equation*}
  (\bS^{1}\times\bS^{2})_{i_{1}\dotsb i_{p+q-1}} : = (\varepsilon_{i_{1}jk}S^{1}_{ji_{2}\dotsb i_{p}}S^{2}_{ki_{p+1} \dotsb i_{p+q-1}})^{s}.
\end{equation*}
This operation is \emph{covariant} for the special orthogonal group $\SO(3)$, but not for the full orthogonal group $\OO(3)$. In that case, for any $g\in \OO(3)$ we have:
\begin{equation*}
  (\rho(g) \bS^{1}) \times (\rho(g)\bS^{2}) = \hrho(g) \left(\bS^{1} \times \bS^{2}\right), \qquad
  (\hrho(g) \bS^{1}) \times (\hrho(g) \bS^{2}) = \rho(g) \left(\bS^{1} \times \bS^{2}\right).
\end{equation*}

\subsection{Covariant operations on polynomials}
\label{subsec:polynomial-counter-parts}

Each covariant tensorial operation between \emph{totally symmetric tensors} has its covariant polynomial counterpart (see~\cite{OKDD2018}). Let $\bS^{1}, \bS^{2}$ be two totally symmetric tensors of respective order $p$, $q$ and let $\rp_{1}, \rp_{2}$ be their respective polynomial counter-part.

\begin{enumerate}
  \item The symmetric tensor product translates into the \emph{ordinary product} between polynomials
        \begin{equation*}\label{eq:SymProduct_Product}
          \varphi(\bS^{1} \odot \bS^{2}) = \rp_{1}\, \rp_{2}.
        \end{equation*}

  \item The generalized cross product translates into \emph{Lie-Poisson bracket} of $\rp_{1}$ and $\rp_{2}$
        \begin{equation*}\label{eq:CrossProduct_Poisson}
          \varphi(\bS^{1}\times\bS^{2}) = \frac{1}{pq} \poisson{\rp_{1}}{\rp_{2}}_\textrm{LP}
        \end{equation*}
        where
        \begin{equation*}
          \poisson{\rp_{1}}{\rp_{2}}_\textrm{LP} : = \det(\xx,\nabla \rp_{1}, \nabla \rp_{2}),
        \end{equation*}
        and $\nabla \rp$ is the gradient of $\rp$;

  \item The symmetric $r$-contraction translates into the \emph{Euclidean transvectant} of order $r$
        \begin{equation*}\label{eq:euclidean-transvectant}
          \varphi((\bS^{1} \overset{(r)}{\cdot} \bS^{2})^s) = \frac{(p-r)!}{p!}\frac{(q-r)!}{q!} \poisson{\rp_{1}}{\rp_{2}}_{r}
        \end{equation*}
        where
        \begin{equation*}
          \poisson{\rp_{1}}{\rp_{2}}_{r} : = \sum_{k_{1}+k_{2}+k_{3} = r} \frac{r!}{k_{1}!k_{2}!k_{3}!}\frac{\partial^r \rp_{1}}{\partial x^{k_{1}}\partial y^{k_{2}}\partial z^{k_{3}}}\frac{\partial^r \rp_{2}}{\partial x^{k_{1}}\partial y^{k_{2}}\partial z^{k_{3}}}.
        \end{equation*}

  \item Let $\bnu$ be a vector. The symmetric $q$-contraction $(\bS \overset{(q)}{\cdot} \bnu^q)^s=\bS \overset{(q)}{\cdot} \bnu^q$ between $\bS^{1}=\bS$ and $\bS^{2}=\bnu^q=\underbrace{\bnu \otimes \bnu \otimes \dotsb \otimes \bnu}_{q\: \mathrm{ times }}$ translates into
        \begin{equation*}
          \varphi(\bS \overset{(q)}{\cdot} \bnu^q) = \frac{(p-q)!}{p!q!} \poisson{\rp}{(\bnu\cdot \xx)^q}_{q}.
        \end{equation*}
        where $\rp_{1} = \rp = \varphi(S)$ and $\rp_{2}=(\bnu\cdot \xx)^q$.
\end{enumerate}

\section{Symmetry group and order-two symmetries}
\label{subsec:symmetries}

Given any representation $\varrho$ of $\OO(3)$ on $\VV$, the \emph{isotropy group} (or \emph{symmetry group}) of a tensor $\bT\in \VV$ is the subgroup
\begin{equation*}
  G_{\bT} : = \set{g\in \OO(3),\quad \varrho(g)\bT = \bT}.
\end{equation*}
In the present work, we focus on \emph{order-two symmetries}, that is on symmetries $g \in G_{\bT}$ such that
\begin{equation*}
  g \ne \Id, \qquad g^{2} = \Id.
\end{equation*}
It is thus useful to recall that there are three types of order-two elements in the group $\OO(3)$:
\begin{enumerate}
  \item the \emph{central symmetry} $-I$;
  \item \emph{plane symmetries}, which are characterized by a unit vector $\bnu$. The symmetry with respect to the plane $\bnu^\perp$ is the orthogonal transformation $\bs(\bnu) : = \Id-2 \bnu \otimes \bnu$;
  \item \emph{order-two rotational symmetries} (axial symmetries), which are rotations by angle $\pi$ around some axis $\langle\bnu\rangle$. Note that $\br(\bnu,\pi) : = -\bs(\bnu)$.
\end{enumerate}t

\begin{defn}[Order-two symmetries]
  Let $\varrho$ be a representation of $\OO(3)$ on $\VV$, $\bnu$ be a unit vector in $\RR^{3}$ and $\bT\in \VV$. Then $\bnu^{\perp}$ is a \emph{symmetry plane} of $\bT$ if
  \begin{equation}\label{eq:Planar_Equation}
    \varrho(\bs(\bnu)) \bT = \bT,
  \end{equation}
  and the axis $\langle\bnu\rangle$ is a \emph{symmetry axis} of $\bT$ if
  \begin{equation}\label{eq:Flip_Equation}
    \varrho(\br(\bnu,\pi))\bT = \bT.
  \end{equation}
\end{defn}

\begin{rem}\label{rem:reduction-to-e3}
  Note that for any $g \in \OO(3)$, $G_{\varrho(g)\bT} = gG_{\bT}g^{-1}$. In particular, given a unit vector $\bnu$, we can choose $g\in \SO(3)$ such that $g\bnu = \ee_3$. Then $\bnu^{\perp}$ is a  symmetry plane of $\bT$ if and only if $\ee_3^{\perp}$ is a symmetry plane of $\varrho(g)\bT$. Similarly, $\langle\bnu\rangle$ is a symmetry axis of $\bT$ if and only if $\langle\ee_3\rangle$ is a symmetry axis of $\varrho(g)\bT$.
\end{rem}

The equations for the existence of an axial symmetry for the natural representation or a plane/axial symmetry for the twisted representation of a tensor $\bT\in \TT^{n}(\RR^{3})$ can be reduced to a condition on $\rho_{n}(\bs(\bnu))\bT$. These conditions are given in Table~\ref{tab:second-order-sym} and are deduced from
\begin{equation*}
  \rho_{n}(\bs(\bnu))  \bT = -\hrho_{n}(\bs(\bnu)) \bT = (-1)^{n} \rho_{n}(\br(\bnu,\pi)) \bT = (-1)^{n}\hrho_{n}(\br(\bnu,\pi)) \bT.
\end{equation*}

\begin{table}[H]
  \begin{center}
    \begin{tabular}{ccc}
      \toprule
                                                                             & Natural representation $\rho$                                         & Twisted representation $\hrho$ \\
      \midrule
      $n$ even                                                               & $\rho_{n}(\br(\bnu,\pi))\bT = \bT \iff \rho_{n}(\bs(\bnu)) \bT = \bT$ &
      $\begin{aligned}
          \hrho_{n}(\bs(\bnu)) \bT = \bT     & \iff \rho_{n}(\bs(\bnu)) \bT = -\bT
          \\
          \hrho_{n}(\br(\bnu,\pi)) \bT = \bT & \iff \rho_{n}(\bs(\bnu)) \bT = \bT
        \end{aligned}$
      \\
      \midrule
      $n$ odd                                                                &
      $\rho_{n}(\br(\bnu,\pi))\bT = \bT \iff \rho_{n}(\bs(\bnu)) \bT = -\bT$ &
      $\begin{aligned}
          \hrho_{n}(\bs(\bnu)) \bT = \bT     & \iff \rho_{n}(\bs(\bnu)) \bT = -\bT
          \\
          \hrho_{n}(\br(\bnu,\pi)) \bT = \bT & \iff \rho_{n}(\bs(\bnu)) \bT = - \bT
        \end{aligned}$
      \\
      \bottomrule
    \end{tabular}
  \end{center}
  \caption{Order-two symmetries}\label{tab:second-order-sym}
\end{table}

\begin{ex}
  Let $\ba$ be a symmetric second-order tensor with three distinct eigenvalues. In some orthonormal basis $(\ee_{i})$ we have:
  \begin{equation*}
    \ba = \begin{pmatrix}
      \lambda_{1} & 0           & 0         \\
      0           & \lambda_{2} & 0         \\
      0           & 0           & \lambda_3
    \end{pmatrix},\quad \lambda_{i}\neq \lambda_j \text{ for } i\neq j.
  \end{equation*}

  Under the representation $\rho_{2}$, it has three symmetry planes whose normals $\bnu = \ee_{i}$ ($i = 1,2,3$) are also symmetry axes and $\ba$ is said to be orthotropic.

  Under the representation $\hrho_{2}$, the situation is more subtle:
  \begin{itemize}
    \item If $\lambda_{1}\lambda_{2}\lambda_3\neq 0$ or $\Pi_{i\neq j} (\lambda_{i}+\lambda_j)\neq 0$ then $\ba$ has three symmetry axes $\langle\bnu\rangle = \langle \ee_{i}\rangle$ ($i = 1,2,3$)  but no symmetry plane;
    \item If $\lambda_{i}\neq 0$ for some $i$ and $\lambda_j+\lambda_k = 0$ with $\lambda_j\neq0$ and $j,k\neq i$, then $\ba$ has also two symmetry planes, of normals $(\ee_j\pm\ee_k)/\sqrt{2}$.
  \end{itemize}
\end{ex}

\section{Reduced algebraic equations for second-order symmetries}
\label{sec:reduced-equations}

Let $\bnu$ be a unit vector and $\bs(\bnu)=\Id-2 \bnu \otimes \bnu$ be the corresponding plane symmetry. The action of $\bs(\bnu)$ on a tensor $\bS\in(\Sym^{n}(\RR^{3}),\rho_{n})$ writes as
\begin{equation*}
  \rho_{n}(\bs(\bnu)) \bS = \sum_{k = 0}^{n} (-2)^k \binom{n}{k} \contsnu{k}{k}
\end{equation*}
and thus, the equation
\begin{equation*}
  \rho_{n}(\bs(\bnu))\bS = \bS
\end{equation*}
leads to a system of $(n+1)(n+2)/2$ polynomial equations in $\bnu$, each of them being homogeneous of degree $2n$. For instance, when $\bS = \ba \in (\Sym^{2}(\RR^{3}),\rho_{2})$ is of order 2, the equations
\begin{equation}\label{eq:rhosaegala}
  \rho_{2}(\bs(\bnu)) \ba = \ba,
\end{equation}
leads to
\begin{equation}\label{eq:rhosaegala_pol}
  \ba-4 \bnu \odot (\ba \cdot \bnu)+4 ( \bnu\cdot \ba \cdot\bnu)\, \bnu \otimes \bnu = \ba ,
\end{equation}
which is a system of six polynomial equations in $\bnu = (x,y,z)$, each of them being of degree $4$.

Now, one can recast the problem \eqref{eq:rhosaegala} in a different way by observing that $\bnu^{\perp}$ is a symmetry plane of $\ba$ if and only if $\bnu$ is an eigenvector of $\ba$. In other words
\begin{equation*}
  \rho_{2}(\bs(\bnu)) \ba = \ba \iff \ba\cdot\bnu = \lambda\bnu \iff (\ba\cdot\bnu)\times \bnu = 0,
\end{equation*}
where $\times$ stands here for the standard cross product between vectors. In this reformulation of the problem, one obtains, however three polynomial equations in $\bnu$, each of them being of degree $2$. One has therefore divided the degree of the equations by 2, compared to~\eqref{eq:rhosaegala}.

\emph{Theorem~\ref{thm:main} is a generalization of this simple observation to any totally symmetric $n$th-order tensor or pseudo-tensor}. The extension to any (pseudo-)tensors and, more generally, to any finite dimensional representation $\varrho$ of the orthogonal group $\OO(3)$ follows, using the harmonic decomposition, and will be illustrated in the next sections.

To formulate our result, we introduce the following notations. Given a unit vector $\bnu$, we define
\begin{equation*}
  \bnu^{k} := \underset{k \, \text{times}}{\underbrace{\bnu \otimes \bnu \otimes \dotsb \otimes \bnu}} \in \Sym^{k}(\RR^{3})
\end{equation*}
as the tensorial product of $k$ copies of the vector $\bnu$. Given $n \ge 1$, we set $r := \lfloor n/2 \rfloor$, $q := \lfloor (n+1)/2 \rfloor$, and define the $(n+1)\times (n+1)$ matrix $B_{n}$ by
\begin{equation}\label{eq:Bij}
  \left( B_{n}\right)_{ij} : =
  \begin{cases}
    \dfrac{(i)!}{(i-j)!} \text{ if } i-j\geq 0, \\
    0 \text{ otherwise},
  \end{cases}
\end{equation}
where $i,j\in [0, n]$. We denote by $B^{1}_{n}$ the $q \times q$ matrix obtained from $B_{n}$ by deleting columns
\begin{equation*}
  0, q+1,q+2, \dotsc ,n,
\end{equation*}
and rows
\begin{equation*}
  0, 1, 3, \dotsc 2r-1,
\end{equation*}
and by $\tilde B^{1}_{n}$ the $(r+1) \times (r+1)$ matrix obtained from $B_{n}$ by deleting columns
\begin{equation*}
  r+1,r+2, \dotsc ,n,
\end{equation*}
and rows
\begin{equation*}
  0, 2, 4, \dotsc ,2q-2.
\end{equation*}
Finally, we set
\begin{equation*}
  \lambda^{(n)}_{j} := (B^{1}_{n})^{-1}_{jq}\quad (j\in [1, q]),
  \quad \text{and} \quad
  \mu^{(n)}_{j} := (\tilde B^{1}_{n})^{-1}_{jr}\quad (j\in [0, r]).
\end{equation*}
The subscripts $i$ and $j$ of the matrix $B^{1}_{n}=((B^{1}_{n})_{ij})$ and  the vector $\lambda^{(n)}=(\lambda^{(n)}_{j})$ varies from 1 to $q$. For matrix $\tilde B^{1}_{n}=((\tilde B^{1}_{n})_{ij})$ and vector $\mu^{(n)}=(\mu^{(n)}_{j})$, $i$ varies from $1$ to $r+1$ and $j$ from $0$ to $r$.

\begin{ex}
  For $n=3$ we get $r=1$, $q=2$ and
  \begin{equation*}
    B_3=\left(
    \begin{array}{cccc}
      1 & 0 & 0 & 0 \\
      1 & 1 & 0 & 0 \\
      1 & 2 & 2 & 0 \\
      1 & 3 & 6 & 6
    \end{array}
    \right)
    ,
    \qquad
    \begin{cases}
      B^{1}_3=
      \left(
      \begin{array}{ccc}
          2 & 2 \\
          3 & 6
        \end{array}
      \right),
      \quad
      \lambda^{(3)}_{1}=-\lambda^{(3)}_{2} =-\frac{1}{3}.
      \\
      \tilde B^{1}_3=\left(
      \begin{array}{cc}
          1 & 1 \\
          1 & 3 \\
        \end{array}
      \right)
      ,
      \quad
      \mu^{(3)}_{0}=-\mu^{(3)}_{1} =-\frac{1}{2}.
    \end{cases}
  \end{equation*}
\end{ex}

\begin{ex}
  For $n=4$ we get $r=q=2$ and
  \begin{equation*}
    B_4=\left(
    \begin{array}{ccccc}
      1 & 0 & 0  & 0  & 0  \\
      1 & 1 & 0  & 0  & 0  \\
      1 & 2 & 2  & 0  & 0  \\
      1 & 3 & 6  & 6  & 0  \\
      1 & 4 & 12 & 24 & 24 \\
    \end{array}
    \right),
    \qquad
    \begin{cases}
      B^{1}_4=\left(
      \begin{array}{cc}
          2 & 2  \\
          4 & 12 \\
        \end{array}
      \right)
      ,
      \qquad
      \lambda^{(4)}_{1}=-\lambda^{(3)}_{2} =-\frac{1}{8}.
      \\
      \tilde B^{1}_4= \left(
      \begin{array}{ccc}
          1 & 1 & 0  \\
          1 & 3 & 6  \\
          1 & 4 & 12 \\
        \end{array}
      \right)
      ,
      \quad
      \mu^{(4)}_{0}=-\mu^{(4)}_{1} =1, \; \mu^{(4)}_{2}= \frac{1}{3}.
    \end{cases}
  \end{equation*}
\end{ex}

\begin{thm}\label{thm:main}
  Let $\Sym^{n}(\RR^{3})$ be the vector space  of totally symmetric tensors of order $n$, and
  \begin{equation*}
    r := \left\lfloor \frac{n}{2} \right\rfloor, \qquad q := \left\lfloor \frac{n+1}{2} \right\rfloor.
  \end{equation*}
  so that $q = r$, if $n$ is even and $q = r+1$, if $n$ is odd.
  \begin{enumerate}
    \item A unit vector $\bnu$ defines a plane/axial symmetry of $\bS\in \Sym^{2r}(\RR^{3})$ for the representation $\rho_{2r}$ or an axial symmetry of $\bS\in \Sym^{2r}(\RR^{3})$ for the representation $\hrho_{2r}$ if and only if
          \begin{equation}\label{eq:symmetric-even-order-tensors}
            \left[\sum_{k = 1}^{q}\frac{n!}{(n-k)!} \, \lambda^{(n)}_{k} \, \contsnu{k-1}{k}\right]\times \bnu = 0, \qquad n = 2r.
          \end{equation}
    \item A unit vector $\bnu$ defines a plane symmetry of $\bS\in \Sym^{2r}(\RR^{3})$ for the representation $\hrho_{2r}$ if and only if
          \begin{equation}\label{eq:symmetric-even-order-pseudo-tensors}
            \sum_{k = 0}^{r} \frac{n!}{(n-k)!}  \, \mu^{(n)}_{k} \, \contsnu{k}{k} = 0, \qquad n = 2r.
          \end{equation}
    \item A unit vector $\bnu$ define a plane/axial symmetry of $\bS\in \Sym^{2r+1}(\RR^{3})$ for the representation $\hrho_{2r+1}$, or an axial symmetry of $\bS\in \Sym^{2r+1}(\RR^{3})$ for the representation $\rho_{2r+1}$ if and only if
          \begin{equation}\label{eq:symmetric-odd-order-pseudo-tensors}
            \left[\sum_{k = 0}^{r} \frac{n!}{(n-k)!}  \, \mu^{(n)}_{k} \, \contsnu{k}{k}\right]\times \bnu = 0, \qquad n = 2r+1.
          \end{equation}
    \item A unit vector $\bnu$ defines a plane symmetry of $\bS\in \Sym^{2r+1}(\RR^{3})$ for the representation $\rho_{2r+1}$ if and only if
          \begin{equation}\label{eq:symmetric-odd-order-tensors}
            \sum_{k = 1}^{q} \frac{n!}{(n-k)!} \, \lambda^{(n)}_{k} \, \contsnu{k-1}{k} = 0, \qquad n = 2r+1.
          \end{equation}
  \end{enumerate}
\end{thm}

The following~\autoref{tab:reduced-equations} recapitulates all these conditions.

\begin{table}[H]
  \begin{center}
    \begin{tabular}{ccc}
      \toprule
              & Natural representation & Twisted representation \\
      \midrule
      $n$ even
              &
      \begin{tabular}{c}
        $\rho_{n}(\bs(\bnu)) \bS = \bS \Leftrightarrow\rho_{n}( \br(\bnu,\pi)) \bS = \bS$ \\
        $\Updownarrow$                                                                    \\
        Eq.~\eqref{eq:symmetric-even-order-tensors}
      \end{tabular}
              &
      \begin{tabular}{l}
        $\bullet$ Plane: $\hrho_{n}(\bs(\bnu))  \bS = \bS \Leftrightarrow  $Eq.~\eqref{eq:symmetric-even-order-pseudo-tensors} \\
        $\bullet$ Axial: $\hrho_{n}(\br(\bnu,\pi))  \bS = \bS \Leftrightarrow  $ Eq.~\eqref{eq:symmetric-even-order-tensors}
      \end{tabular}                                \\
      \midrule
      $n$ odd &
      \begin{tabular}{l}
        $\bullet$ Plane: $\rho_{n}(\bs(\bnu)) \bS = \bS \Leftrightarrow $ Eq.~\eqref{eq:symmetric-odd-order-tensors} \\
        $\bullet$ Axial: $\rho_{n}(\br(\bnu,\pi)) \bS = \bS \Leftrightarrow $ Eq.~\eqref{eq:symmetric-odd-order-pseudo-tensors}
      \end{tabular}
              &
      \begin{tabular}{c}
        $\hrho_{n}(\bs(\bnu)) \bS = \bS \Leftrightarrow \hrho_{n}(\br(\bnu,\pi)) \bS = \bS$ \\
        $\Updownarrow$                                                                      \\
        Eq.~\eqref{eq:symmetric-odd-order-pseudo-tensors}
      \end{tabular}
      \\
      \bottomrule
    \end{tabular}
  \end{center}
  \caption{Reduced algebraic equations for second-order symmetries.}\label{tab:reduced-equations}
\end{table}

Before providing a proof of theorem~\ref{thm:main}, we explicit these conditions in certain cases.

\begin{ex}\label{ex:even-order-tensors}
  Equation \eqref{eq:symmetric-even-order-tensors} defining a plane/axial symmetry of a totally symmetric tensor or an axial symmetry of a totally symmetric pseudo-tensor
  writes, for orders $n = 2r\leq 10$,
  \begin{description}
    \item[$n = 2$]
          \begin{equation*}
            \left[\bS\cdot \bnu\right] \times \bnu = 0;
          \end{equation*}
    \item[$n = 4$]
          \begin{equation*}
            \left[\bS\cdot \bnu -3 \bnu \odot (\bnu\cdot \bS\cdot \bnu) \right]\times \bnu = 0;
          \end{equation*}
    \item[$n = 6$]
          \begin{equation*}
            \left[\bS\cdot \bnu -5 \contsnu{}{2} + \dfrac{20}{3} \contsnu{2}{3} \right]\times \bnu = 0;
          \end{equation*}
    \item[$n = 8$]
          \begin{equation*}
            \left[\bS\cdot \bnu -7 \contsnu{}{2} + \dfrac{84}{5} \contsnu{2}{3}-14 \contsnu{3}{4} \right]\times \bnu = 0;
          \end{equation*}
    \item[$n = 10$]
          \begin{multline*}
            \left[\bS\cdot \bnu -9 \contsnu{}{2} + \dfrac{216}{7} \contsnu{2}{3} \right. \\
              \left. - 48\contsnu{3}{4} + \dfrac{144}{5}\contsnu{4}{5} \right]\times \bnu = 0.
          \end{multline*}
  \end{description}
\end{ex}

\begin{ex}\label{ex:even-order-pseudo-tensors}
  Equation~\eqref{eq:symmetric-even-order-pseudo-tensors} defining a plane symmetry of a totally symmetric pseudo-tensor
  writes, for orders $n = 2r\leq 10$,
  \begin{description}
    \item[$n = 2$]
          \begin{equation*}
            \bS-2\,\bnu \odot  (\bS \cdot \bnu) = 0;
          \end{equation*}
    \item[$n = 4$]
          \begin{equation*}
            \bS -4\, \bnu \odot \left(\bS\cdot \bnu\right) + 4\bnu \odot \bnu \odot \left(\bnu\cdot \bS\cdot \bnu\right) = 0;
          \end{equation*}
    \item[$n = 6$]
          \begin{equation*}
            \bS -6\, \bnu \odot \left(\bS\cdot \bnu\right) + 12\contsnu{2}{2}-8\contsnu{3}{3} = 0;
          \end{equation*}
    \item[$n = 8$]
          \begin{equation*}
            \bS -8\, \bnu \odot \left(\bS\cdot \bnu\right) + 24\contsnu{2}{2}-32 \contsnu{3}{3} + 16\contsnu{4}{4} = 0;
          \end{equation*}
    \item[$n = 10$]
          \begin{multline*}
            \bS -10\, \bnu \odot \left(\bS\cdot \bnu\right) + 40\contsnu{2}{2} - 80\contsnu{3}{3}  \\
            + 80\contsnu{4}{4} - 32\contsnu{5}{5} = 0.
          \end{multline*}
  \end{description}
\end{ex}

\begin{ex}\label{ex:odd-order-pseudo-tensors}
  Equation~\eqref{eq:symmetric-odd-order-pseudo-tensors} defining a plane/axial symmetry of a totally symmetric pseudo-tensor or an axial symmetry of a totally symmetric tensor
  writes, for orders $n = 2r+1<10$,
  \begin{description}
    \item[$n = 1$]
          \begin{equation*}
            \bS\times \bnu = 0;
          \end{equation*}
    \item[$n = 3$]
          \begin{equation*}
            \left[\bS -3 \bnu \odot \left(\bS\cdot \bnu\right)\right]\times \bnu = 0;
          \end{equation*}
    \item[$n = 5$]
          \begin{equation*}
            \left[\bS -5\, \bnu \odot \left(\bS\cdot \bnu\right) + \frac{20}{3}\, \bnu \odot \bnu \odot \left(\bnu\cdot \bS\cdot \bnu\right)\right]\times \bnu = 0;
          \end{equation*}
    \item[$n = 7$]
          \begin{equation*}
            \left[\bS -7\, \bnu \odot \left(\bS\cdot \bnu\right) + \frac{84}{5}\, \contsnu{2}{2} -14 \, \contsnu{3}{3} \right]\times \bnu = 0;
          \end{equation*}
    \item[$n = 9$]
          \begin{equation*}
            \left[\bS -9\, \bnu \odot \left(\bS\cdot \bnu\right) + \frac{216}{7}\,  \contsnu{2}{2} - 48 \,  \contsnu{3}{3}
              + \frac{144}{5}\,  \contsnu{4}{4} \right]\times \bnu = 0.
          \end{equation*}
  \end{description}
\end{ex}

\begin{ex}\label{ex:odd-order-tensors}
  Equation~\eqref{eq:symmetric-odd-order-tensors} defining a plane symmetry of a totally symmetric tensor writes, for orders $n = 2r+1<10$,
  \begin{description}
    \item[$n = 1$]
          \begin{equation*}
            \bS\cdot \bnu = 0;
          \end{equation*}
    \item[$n = 3$]
          \begin{equation*}
            \bS\cdot \bnu -2\, \bnu \odot \left(\bnu\cdot \bS\cdot \bnu\right) = 0;
          \end{equation*}
    \item[$n = 5$]
          \begin{equation*}
            \bS\cdot \bnu -4\, \bnu \odot \left(\bnu\cdot \bS\cdot \bnu\right)+ 4\,  \bnu \odot \bnu \odot\left(\left(\bnu\cdot \bS\cdot \bnu\right)\cdot \bnu\right) = 0;
          \end{equation*}
    \item[$n = 7$]
          \begin{equation*}
            \bS\cdot \bnu -6\,  \bnu \odot \left(\bnu\cdot \bS\cdot \bnu\right)+ 12\,  \contsnu{2}{3} - 8\,  \contsnu{3}{4} = 0;
          \end{equation*}
    \item[$n = 9$]
          \begin{equation*}
            \bS\cdot \bnu -8\,  \bnu \odot \left(\bnu\cdot \bS\cdot \bnu\right)+ 24\, \contsnu{2}{3}-32\contsnu{3}{4}
            + 16\, \contsnu{4}{5} = 0.
          \end{equation*}
  \end{description}
\end{ex}

\begin{proof}[Proof of theorem~\ref{thm:main}]
  Let $\bS\in \Sym^{n}(\RR^{3})$ be a totally symmetric tensor of order $n \ge 1$ and let $\rp = \varphi(\bS)\in \Pn{n}(\RR^{3})$ be the corresponding homogeneous polynomial of degree $n$ (see~\autoref{subsec:tot-sym-and-hom-pol}). Then the tensor $\bS$ is invariant under some transformation $g \in \OO(3)$ if and only if
  \begin{equation*}
    (\varrho(g)\rp) (\xx) = \rp(\xx), \qquad \forall \xx \in \RR^{3},
  \end{equation*}
  where $\varrho$ stands either for $\rho$ or $\hrho$. We suppose now that $g = \bs(\bnu)$. Without loss of generality, we can choose an orthonormal basis $(\ee_{1},\ee_{2},\ee_3)$ such that $\bnu = \ee_3$ but then
  \begin{gather*}
    \rho(\bs(\bnu)) \rp = \rp \iff \text{$\rp$ is $z$-even} \quad (\rp(x,y,-z) = \rp(x,y,z));
    \\
    \rho(\bs(\bnu)) \rp = -\rp \iff \text{$\rp$ is $z$-odd} \quad (\rp(x,y,-z) = -\rp(x,y,z)).
  \end{gather*}
  We will now use the necessary and sufficient conditions obtained in \autoref{app:odd-even-polynomials}, namely theorems~\ref{thm:z-even-CNS} and~\ref{thm:z-odd-CNS} which characterize $z$-even or odd homogeneous polynomials. We observe that the scalar product $\bnu \cdot \xx$ and the \emph{Euclidean transvectant} $\poisson{\rp}{(\bnu \cdot \xx)^{k}}_{k}$ defined in~\autoref{subsec:polynomial-counter-parts} write as follows when $\bnu = \ee_3$.
  \begin{align*}
    \bnu \cdot \xx                          & = z ,
    \\
    \poisson{\rp}{(\bnu \cdot \xx)^{k}}_{r} & =  \frac{k!}{(k-r)!} z^{k-r} {\partial_{z}}^{r}\rp, \quad \text{for} \quad k \ge r.
  \end{align*}
  We are thus lead to introduce the following linear operators:
  \begin{equation*}
    \mathcal{L}^{\bnu}_{n}(\rp) : = \sum_{k = 1}^{q} \frac{1}{k!} \lambda^{(n)}_{k} (\bnu \cdot \xx)^{k-1} \poisson{\rp}{(\bnu \cdot \xx)^{k}}_{k}
  \end{equation*}
  and
  \begin{equation*}
    \mathcal{K}^{\bnu}_{n}(\rp) : = \sum_{k = 0}^{r} \frac{1}{k!} \mu^{(n)}_{k} (\bnu \cdot \xx)^{k} \poisson{\rp}{(\bnu \cdot \xx)^{k}}_{k}
  \end{equation*}
  which allows us to formulate conditions which are independent of the particular choice of an orthonormal basis.

  If $n$ is even, then:
  \begin{align*}
    \rho(\bs(\bnu)) \rp = \rp  & \iff  \poisson{\mathcal{L}^{\bnu}_{n}(\rp)}{(\bnu \cdot \xx)}_{LP} = 0,
    \\
    \rho(\bs(\bnu)) \rp = -\rp & \iff  \mathcal{K}^{\bnu}_{n}(\rp) = 0.
  \end{align*}

  If $n$ is odd, then:
  \begin{align*}
    \rho(\bs(\bnu)) \rp = \rp \iff  & \mathcal{L}^{\bnu}_{n}(\rp) = 0,
    \\
    \rho(\bs(\bnu)) \rp = -\rp \iff & \poisson{\mathcal{K}^{\bnu}_{n}(\rp)}{(\bnu \cdot \xx)}_{LP} = 0.
  \end{align*}

  Now, to obtain these criteria in tensorial notations, we use the tensorial counterparts defined in~\autoref{subsec:polynomial-counter-parts} and check that $\mathcal{L}^{\bnu}_{n}(\rp)$ translate as
  \begin{equation*}
    \bL^{\bnu}_{n}(\bS) := \sum_{k = 1}^{q} \frac{n!}{(n-k)!} \, \lambda^{(n)}_{k} \, \contsnu{k-1}{k},
  \end{equation*}
  and $\mathcal{K}^{\bnu}_{n}(\rp)$ as
  \begin{equation*}
    \bK^{\bnu}_{n}(\bS) := \sum_{k = 0}^{r} \frac{n!}{(n-k)!} \, \mu^{(n)}_{k} \, \contsnu{k}{k}.
  \end{equation*}

  Finally, to achieve the proof, we use the formulas in~\autoref{tab:second-order-sym}.
\end{proof}

\section{Application to fabric tensors of directional data}
\label{sec:fabric-tensors}

A directional density $D(\nn)$, related to any possible 3D direction $\nn$, refers to a scalar property defined directionally  in a continuous manner at the Representative Volume Element scale of continuum mechanics. For example $D(\nn)$ may represent the directional density of spatial contacts and grains orientations within granular materials~\cite{Oda1982,RKM2009,RDAR2012,LD2015,KR2016}.
It may represent the directional description of crack density (representative of spatial arrangement, orientation and geometry of the cracks present at the microscale~\cite{Lad1983,Lad1995,Kan1984,Ona1984,LK1993,Kac1993,TNS2001,CW2010}).
It may also represent the directional (tensorial) description of microstructure degradation by rafting in single crystal superalloys at high temperature~\cite{DMC2017,CDC2018}. Comprehensive descriptions of rafting phenomenon can be found in~\cite{KSKGE2001,IKH+2003}.

A model for directional density is given by a \emph{homogeneous} polynomial $D(\nn)$ of even order $n = 2r$ (as we must have $D(\nn) = D(-\nn)$ over a Representative Volume Element). This homogeneous polynomial in $\nn$ corresponds to a totally symmetric tensor $\bF$ --- a so-called \emph{fabric tensor}~\cite{Kan1984} --- of even order $n = 2r$, and conversely (see~\autoref{subsec:tot-sym-and-hom-pol}):
\begin{equation}\label{eq:DefF}
  D(\nn) =  \bF(\nn, \nn, \dotsc, \nn)= \bF \overset{(n)}{\cdot}  \nn^{n},
  \qquad
  \norm{\nn} = 1,
\end{equation}
where the contraction $\bF \overset{(n)}{\cdot}  \nn^{n}$ is the scalar product between the two $n^{th}$ order symmetric tensors $\bF$ and $\nn^n=\nn \otimes \nn \dotsc \otimes \nn$.

The fabric tensor $\bF$ can be determined from the least square error approximation of an experimental (measured) density distribution $D^\textrm{exp}(\nn)$. As $\bF$ is totally symmetric, equation~\eqref{eq:symmetric-even-order-tensors} in theorem~\ref{thm:main} determining the unit normals $\bnu$ to all symmetry planes directly applies to fabric tensors of any order $n = 2r$ with $\bS=\bF\in \Sym^{2r}$ and representation $\rho_{2r}$ (see Example~\ref{ex:even-order-tensors} for the cases $n = 2$ to 10).

\section{Application to Elasticity tensors}
\label{sec:elasticity-tensors}

In Linear Elasticity theory, we have the relationship
\begin{equation*}
  \bsigma = \bC: \bepsilon,
  \qquad
  \sigma_{ij} = C_{ijkl} \epsilon_{kl}
\end{equation*}
between symmetric stress and strain tensors $\bsigma, \bepsilon \in \Sym^{2}(\RR^{3})$ equipped with the representation $\rho_{2}$. The Elasticity tensor $\bC\in \Ela$ has the index symmetries
$C_{ijkl} = C_{ijlk} = C_{klij}$ and has thus two independent traces, the \emph{dilatation} tensor $\bd$ and \emph{Voigt's} tensor $\bv$:
\begin{equation*}
  \bd : = \tr_{12} \bC = \bC:\Id \quad (d_{ij} = C_{ppij}= C_{ijpp}),
  \qquad
  \bv : = \tr_{13}\bC \quad (v_{ij} = C_{pipj}),
\end{equation*}
which are symmetric second-order tensors.

The reduced equations determining the symmetry planes of a given Elasticity tensor are presented in the following theorem~\ref{thm:elasticity}. Its proof relies on its \emph{harmonic decomposition} (see~\autoref{subsec:harmonic-decomposition}). More precisely, let $\bC^{s}$ be the totally symmetric part of $\bC$
\begin{equation*}
  (\bC^{s})_{ijkl} = \frac{1}{3}(C_{ijkl}+C_{ikjl}+C_{iljk}),
\end{equation*}
$\bb' = \bb-\frac{1}{3}\tr \bb\, \Id$ be the deviatoric part of the second-order symmetric tensor
$\bb = 2(\bd-\bv)$ and $\beta = \frac{1}{6}\tr (\bd-\bv)$. We have then
\begin{equation*}
  \bC = \bC^{s} + \Id \otimes_{(2,2)} \! \bb'+  \beta \,\Id \otimes_{(2,2)}\! \Id,
\end{equation*}
where the Young-symmetrized tensor product $\otimes_{(2,2)}$ of two symmetric second-order tensors $\by, \bz$ is defined as
\begin{equation*}
  \by \otimes_{(2,2)}\!\bz = \frac{1}{3} \big( \by \otimes \bz +  \bz \otimes \by
  - \by \otimesbar \bz - \bz \otimesbar \by \big),\quad (\by \otimesbar \bz)_{ijkl} : = \frac{1}{2} (y_{ik}z_{jl} + y_{il}z_{jk}).
\end{equation*}
This decomposition
\begin{equation}\label{eq:EqIso1}
  \bC\mapsto (\bC^{s},\bb', \beta)\in \Sym^{4}(\RR^{3})\oplus \HH^{2}(\RR^{3})\oplus \HH^{0}(\RR^{3})
\end{equation}
is \emph{equivariant}~\cite{Bac1970}, where $\Sym^{4}(\RR^{3}), \HH^{2}(\RR^{3})$ and $\HH^{0}(\RR^{3})$ are respectively equipped with the representations $\rho_{4}$, $\rho_{2}$ and $\rho_{0}$. We have the following result.

\begin{thm}\label{thm:elasticity}
  Let $\bC$ be an Elasticity tensor, $\bC^{s}$ its totally symmetric part and $\bb = 2(\bd-\bv)$.
  The following conditions are equivalent:
  \begin{enumerate}
    \item  the plane $\bnu^{\perp}$ is a symmetry plane of $\bC$;
    \item  the axis $\langle\bnu\rangle$ is a symmetry axis of $\bC$;
    \item  the following equations are satisfied
          \begin{equation*}
            \left[\bC^{s} \cdot \bnu -3 \bnu\odot \left(\bnu\cdot\bC^{s} \cdot \bnu\right)\right]\times \bnu = 0 \quad \text{and} \quad \left( \bb\cdot \bnu\right)\times \bnu = 0;
          \end{equation*}
    \item the following equations are satisfied
          \begin{equation*}
            \left[\bC^{s} \cdot \bnu -3 \bnu\odot \left(\bnu\cdot\bC^{s} \cdot \bnu\right)\right]\times \bnu = 0 \quad \text{and} \quad \left(\bd \cdot \bnu\right)\times \bnu = 0.
          \end{equation*}
  \end{enumerate}
\end{thm}

\begin{proof}
  We will prove that $(1) \implies (2) \implies (3) \implies (4)\implies (1)$.

  $(1) \implies (2)$ follows directly from remark~\ref{rem:moinsId} (or Table~\ref{tab:second-order-sym}) with $n = 4$ and standard action $\rho_4$.

  $(2) \implies (3)$: since~\eqref{eq:EqIso1} is equivariant, we have
  \begin{equation*}
    \rho_4(g) \bC = \bC \iff \rho(g) (\bC^{s},\bb', \beta) = (\rho_4(g)\bC^{s},\rho_{2}(g)\bb', \beta) = (\bC^{s},\bb', \beta).
  \end{equation*}
  Hence, $\bnu^{\perp}$ is a symmetry plane of $\bC$ if and only if $\bnu^{\perp}$ is a symmetry plane of $\bC^{s}$ and $\bb$ and the result follows from equations~\eqref{eq:symmetric-even-order-tensors} in theorem~\ref{thm:main}.

  $(3) \implies (4)$: if the first equation in $(3)$ is satisfied, then $\bnu^\perp$ is a symmetry plane of $\bC^s$ by theorem~\ref{thm:main}. Hence, $\bnu^\perp$ is a symmetry plane of its second-order covariant $\tr(\bC^s)$ and thus
  \begin{equation*}
    \tr(\bC^s)\times \bnu = 0.
  \end{equation*}
  But
  \begin{equation*}
    \tr (\bC^s) = \frac{1}{3}(\bd+2 \bv) =  \bd - \frac{1}{3}\bb,
  \end{equation*}
  and we deduce that $\bd\times \bnu = 0$.

  $(4) \implies (1)$: if the equations in $(4)$ are satisfied, then, $\bnu^\perp$ is a symmetry plane of both $\bC^s$ and $\bd$ by theorem~\ref{thm:main}. It is thus a symmetry plane of $\bb$ (and of $\bb'$) since
  \begin{equation*}
    \bb = 3(\bd - \tr (\bC^s)).
  \end{equation*}
  By~\eqref{eq:EqIso1}, we deduce then that $\bnu^\perp$ is a symmetry plane of $\bC$.
\end{proof}

Theorem~\ref{thm:elasticity} is an important improvement to Cowin's theorem~\ref{thm:CM1987}. It determines the normals to all the symmetry planes of an Elasticity tensor $\bC$ and this for any symmetry class. It is moreover constructive.
\begin{itemize}
  \item The first equation $\left[\bC^{s}\cdot \bnu -3 \bnu\odot \left(\bnu\cdot\bC^{s}\cdot \bnu\right)\right]\times \bnu = 0$ is polynomial of degree 4 in $\bnu$. It is stored in a third-order tensor;
  \item The second equation $\left[ \bd \cdot \bnu\right]\times \bnu = 0$ (or $\left[ \bb \cdot \bnu\right]\times \bnu = 0$) is polynomial of degree 2 in $\bnu$ and is stored in a vector.
\end{itemize}
Note finally that, in contrast with Cowin--Mehrabadi theorem, these equations are still useful in the cubic class, for which $\bd' = (\bC:\Id)' = 0$, $\bv' = 0$.

\section{Application to Piezo-electricity tensors}
\label{sec:Piezo-electricity-tensors}

In linear Piezo-electricity, an electric field represented by a vector $\vec E \in \TT^{1}(\RR^{3})$ (with representation $\rho_{1}$) generates a strain, represented by a symmetric second-order tensor $\bepsilon \in \Sym^{2}(\RR^{3})$ (with representation $\rho_{2}$). At vanishing stress, the linear relationship
\begin{equation*}
  \bepsilon = \bP\cdot \vec E,
  \qquad
  \epsilon_{ij} = P_{ijk}E_k
\end{equation*}
defines the Piezo-electricity third-order tensor $\bP\in \Piez \subset \TT^{3}(\RR^{3})$ (with representation $\rho_{3}$)~\cite{EM1990}, with index symmetry $P_{ijk} = P_{jik}$.

As in the preceding section, to apply theorem~\ref{thm:main}, we need to decompose the vector space $\Piez$ into totally symmetric tensor spaces in an equivariant manner~\cite{JCB1978,GW2002}. We will write
\begin{equation*}
  \bP = \bP^{s}+\Id \otimes \ww-\Id\odot \ww +\frac{1}{3}\left(\ba\cdot\lc+{}^{t}(\ba\cdot\lc)\right),
\end{equation*}
where $\bP^{s}$ is totally symmetric part of $\bP$,
\begin{equation*}
  (\bP^{s})_{ijk} = \frac{1}{3}(P_{ijk}+P_{ikj}+P_{kji}),
\end{equation*}
\begin{equation*}
  \ww := \frac{3}{4} \left(\tr_{12} \bP-\tr (\bP^{s})\right),
  \qquad
  \ba := \left(\bP:\lc\right)^{s},
\end{equation*}
and the transpose ${}^t \bT$ is on the left two subscripts (${}^t T_{ijk} : = T_{jik}$). This decomposition induces an equivariant isomorphism
\begin{equation*}
  \bP\in \Piez\mapsto (\bP^{s},\ba,\ww)\in \Sym^{3}(\RR^{3})\oplus \HH^{2}(\RR^{3})\oplus \HH^{1}(\RR^{3}),
\end{equation*}
where $\Sym^{3}(\RR^{3})$, $\HH^{2}(\RR^{3})$ and $\HH^{1}(\RR^{3})$ are respectively equipped with the representations $\rho_{3}$, $\hrho_{2}$ and $\rho_{1}$.

Following the same proof as for theorem~\ref{thm:elasticity} (with the use of reduced equations in theorem~\ref{thm:main} for third-order tensor $\bP^s$, second-order pseudo-tensor $\ba$ and first order tensor $\ww$), we obtain reduced equations for the existence of second-order symmetries of Piezo-electricity tensors.

\begin{thm}\label{thm:piezo-electricity}
  Let $\bP\in \Piez$ be a Piezo-electricity tensor, $\bP^{s}$ its totally symmetric part,
  \begin{equation*}
    \ww : = \frac{3}{4} \left(\tr_{12} \bP-\tr \bP^{s}\right),\quad \text{and} \quad \ba : = \left(\bP:\lc\right)^{s},
  \end{equation*}
  where $\lc$ is the Levi-Civita tensor. Let $\bnu$ be a unit vector.
  \begin{enumerate}
    \item The plane $\bnu^{\perp}$ is a symmetry plane of $\bP$ if and only if
          \begin{equation}\label{eq:Planar_Piezo}
            \begin{cases}
              \bP^{s} \cdot \bnu
              -2 \bnu\odot \left(\bnu\cdot\bP^{s} \cdot \bnu\right) = 0, \\
              \ba-2\bnu\odot (\ba\cdot \bnu) = 0,                        \\
              \ww\cdot \bnu = 0.
            \end{cases}
          \end{equation}
    \item The axis $\langle \bnu\rangle$ is a symmetry axis of $\bP$ if and only if
          \begin{equation}\label{eq:Axis_Piezo}
            \begin{cases}
              \left[\bP^{s} -3 \bnu \odot (\bP^{s}\cdot \bnu)\right]\times \bnu = 0 , \\
              \left[\ba\cdot\bnu\right]\times \bnu = 0,                               \\
              \ww\times  \bnu = 0.
            \end{cases}
          \end{equation}
  \end{enumerate}
\end{thm}

Theorem~\ref{thm:piezo-electricity} is constructive. The sets of equations~\eqref{eq:Planar_Piezo} and~\eqref{eq:Axis_Piezo} determine all second-order symmetries of a Piezo-electricity tensor $\bP$, independently of their symmetry class. In both~\eqref{eq:Planar_Piezo} and~\eqref{eq:Axis_Piezo}, the first equation is polynomial of degree 3 in $\bnu$ (stored in a third-order tensor), the second eqaution is polynomial of degree 2 in $\bnu$ (stored in a second-order tensor) and the third one is linear in $\bnu$ (stored in a vector).

\section{Application to Piezo-magnetism tensors}
\label{sec:Piezo-magnetism-tensors}

In linear Piezo-magnetism, a magnetization represented by a pseudo-vector $\vec M\in \TT^{1}(\RR^{3})$ (with representation $\hrho_{1}$) generates a strain (a symmetric second-order tensor $\bepsilon \in \Sym^{2}(\RR^{3})$ with representation $\rho_{2}$) through a linear relation. At vanishing stress and around initial magnetization $\vec M^{0}$, this writes
\begin{equation*}
  \bepsilon = \bPi\cdot (\vec M-\vec M^{0}),
  \qquad
  \epsilon_{ij} = \Pi_{ijk} (M_{k}-M_{k}^{0}),
  \qquad
  \Pi_{ijk} = \Pi_{jik}
\end{equation*}
where $\bPi \in \Magn \subset \TT^{3}(\RR^{3})$ with representation $\hrho_{3}$ is the \emph{Piezo-magnetism third-order pseudo-tensor}.

The subtlety here is that the correct geometry of the piezo-magnestism tensor (and this stands also for the piezo-electricity tensor) is not the full orthogonal group $\OO(3)$ but the full Lorentz group $\OO(1,3)$. In particular the \emph{time-reversal symmetry} plays an important role for the piezo-magnetism tensor, which changes sign when a time-reversal symmetry is applied (see~\cite{EM1990,KE1990,WG2004}). Thus it may be useful, not only to look for plane/axial symmetries of $\bPi$, but also for combination of a symmetry plane with the central symmetry $-I$ (which acts as the time-reversal symmetry on $\bPi$). Combined with observations in~\autoref{tab:second-order-sym}, and in particular the fact that $\hrho_{3}(\bs(\bnu)) = \hrho_{3}(\br(\bnu))$, we are lead to state the following result which is a consequence of theorem~\ref{thm:main}.

\begin{thm}\label{thm:piezo-magnetism}
  Let $\bPi \in \Magn \subset \TT^{3}(\RR^{3})$ with representation $\hrho_{3}$ be a Piezo-magnetism tensor, $\bPi^{s}$ its totally symmetric part,
  \begin{equation*}
    \ww := \frac{3}{4} \left(\tr_{12} \bPi-\tr \bPi^{s}\right),\quad \text{and} \quad \ba := \left(\bPi:\lc\right)^{s}.
  \end{equation*}
  where $\lc$ is the Levi-Civita tensor. Let $\bnu$ be a unit vector. Then
  \begin{enumerate}

    \item $\hrho_3(\bs(\bnu))\bPi = \bPi$ if and only if
          \begin{equation}\label{eq:Axis_Magn}
            \begin{cases}
              \left[\bPi^{s} -3 \bnu \odot (\bPi^{s}\cdot \bnu) \right]\times \bnu = 0, \\
              \left[\ba\cdot\bnu\right]\times \bnu = 0,                                 \\
              \ww \times \bnu = 0.
            \end{cases}
          \end{equation}

    \item $\hrho_3(\bs(\bnu))\bPi = -\bPi$ if and only if
          \begin{equation}\label{eq:Planar_Magn}
            \begin{cases}
              \bPi^{s} \cdot \bnu -2 \bnu\odot \left(\bnu\cdot\bPi^{s} \cdot \bnu\right) = 0, \\
              \ba-2\bnu\odot (\ba\cdot \bnu) = 0,                                             \\
              \ww \cdot \bnu = 0.
            \end{cases}
          \end{equation}
  \end{enumerate}
\end{thm}

\section{Conclusion}
\label{sec:conclusion}

By exploiting the link between tensorial and polynomial representations of the orthogonal group $\OO(3)$, we  have formulated necessary and sufficient conditions which characterize plane and axial symmetries of any totally symmetric (pseudo-)tensors of order $n \ge 1$. These results are stated in theorem~\ref{thm:main} and detailed for orders $n\leq 10$ in examples~\ref{ex:even-order-tensors} to~\ref{ex:odd-order-tensors}.

The proofs of these results, emphasize the practical interest --- in terms of effective calculus --- of explicit translations between covariant operations on totally symmetric tensors and those on homogeneous polynomials. Among them, in particular, the \emph{generalized cross product} \eqref{eq:cross-product}, which generalizes the vector product in $\RR^{3}$ for totally symmetric tensors of any order.

These results are then extended to any tensor $\bT$, using the harmonic decomposition~\cite{Sch1951,Spe1970}. This decomposition being equivariant, a symmetry of the full tensor has to be checked on each factor. This is done using theorem~\ref{thm:main} for each component of the harmonic decomposition, which is a symmetric tensor. The specific cases of \emph{Elasticity tensors}, \emph{Piezo-electricity tensors} and \emph{Piezo-magnetism pseudo-tensors} have been detailed.

Necessary and sufficient conditions for the existence of order-two symmetries of a given tensor constitute an important first step towards the effective determination of the tensor full symmetry group. These equations solve the problem for any Elasticity tensor~\cite{CVC2001}. But, as pointed out in~\cite{CVC2001}, this is only a first step, as the symmetry group of an odd-order tensor is not generated by its order-two symmetries, in general (this applies, in particular to the Piezo-electricity tensor).

\appendix

\section{Homogeneous polynomials even or odd in one variable}
\label{app:odd-even-polynomials}

In this appendix, we will formulate necessary and sufficient conditions for an homogeneous polynomial $\rp$ of degree $n$ in three variables $(x,y,z)$ to be even or odd in $z$. To do so, we introduce, for $k \ge 1$, the following differential operator
\begin{equation*}
  \mathcal{D}_{k} : = z^{k-1}{\partial_{z}}^{k}.
\end{equation*}
Then, if
\begin{equation*}
  \rp = a_{0}(x,y) + a_{1}(x,y)z + \dotsb + a_{n}(x,y)z^{n},
\end{equation*}
we get
\begin{equation*}
  \mathcal{D}_{k}(\rp) = k! \, a_{k}(x,y)z^{k-1} + \dotsb + \frac{n!}{(n-k)!}a_{n}(x,y)z^{n-1}.
\end{equation*}

We will first look for a necessary and sufficient condition for a polynomial $\rp$ to be $z$-even. In order to achieve this goal, we will start by showing that we can find a linear combination of the $\mathcal{D}_{k}(\rp)$ which cancels all the even coefficients $a_{2i}$ but not the odd coefficients $a_{2i-1}$.

\begin{lem}\label{lem:Ln-even-operator}
  Let $n \ge 1$ and set $q = \lfloor \dfrac{n+1}{2}\rfloor$, $r = \lfloor \dfrac{n}{2}\rfloor$, so that $q = r$, if $n$ is even and $q = r+1$, if $n$ is odd. Then, there exists a unique $q$-tuple $(\lambda^{(n)}_{1}, \dotsc, \lambda^{(n)}_{q})$ of rational numbers, solution of the equations (in the unknowns $(\lambda_{1}, \dotsc, \lambda_{q})$)
  \begin{equation}\label{eq:lambda-equations}
    \sum_{k = 1}^{2i} \frac{(2i)!}{(2i-k)!}\lambda_{k} = 0, \quad \text{for} \quad i = 1, \dotsc , q-1 \quad \text{and} \quad \sum_{k = 1}^{q} \frac{n!}{(n-k)!} \lambda_{k} = 1.
  \end{equation}
  Moreover, if we define the linear operator
  \begin{equation}\label{eq:def-Ln}
    \mathcal{L}_{n} : = \sum_{k = 1}^{q} \lambda^{(n)}_{k} \mathcal{D}_{k},
  \end{equation}
  then, for every homogeneous polynomial $\rp$ of degree $\le n$, we have
  \begin{equation*}
    \mathcal{L}_{n}(\rp) = \sum_{i = 1}^{r} \alpha^{(n)}_{i}a_{2i-1}(x,y)z^{2i-2} + a_{n}z^{n-1},
  \end{equation*}
  where
  \begin{equation*}
    \alpha^{(n)}_{i} = \sum_{k = 1}^{2i-1} \dfrac{(2i-1)!}{(2i-1-k)!} \lambda^{(n)}_{k} \ne 0, \qquad 1 \le i \le r.
  \end{equation*}
\end{lem}

\begin{ex}\label{ex:Ln}
  \begin{itemize}
    \item For $n = 1$, we have $r = 0$, $q = 1$ and
          \begin{equation*}
            \mathcal{L}_{1}(\rp) = \partial_{z}\rp = a_{1}.
          \end{equation*}
    \item For $n = 2$, we have $r = 1$, $q = 1$ and
          \begin{equation*}
            \mathcal{L}_{2}(\rp) = \frac{1}{2}\partial_{z}\rp = \frac{1}{2} a_{1} + a_{2}z.
          \end{equation*}
    \item For $n = 3$, we have $r = 1$, $q = 2$ and
          \begin{equation*}
            \mathcal{L}_{3}(\rp) = -\frac{1}{3}\partial_{z}\rp + \frac{1}{3}z {\partial_{z}}^{2}\rp = -\frac{1}{3}a_{1} + a_{3}z^{2}.
          \end{equation*}
    \item For $n = 4$, we have $r = 2$, $q = 2$ and
          \begin{equation*}
            \mathcal{L}_{4}(\rp) = -\frac{1}{8}\partial_{z}\rp + \frac{1}{8}z {\partial_{z}}^{2}\rp = -\frac{1}{8}a_{1} + \frac{3}{8}a_{3}z^{2} + a_{4}z^{3}.
          \end{equation*}
    \item For $n = 5$, we have $r = 2$, $q = 3$ and
          \begin{equation*}
            \mathcal{L}_{5}(\rp) = \frac{1}{5} \partial_{z}\rp -\frac{1}{5}  z {\partial_{z}}^{2}\rp+ \frac{1}{15}  z^{2} {\partial_{z}}^{3}\rp = \frac{1}{5} a_{1} -\frac{1}{5} a_{3}z^{2} + a_{5}z^{4}.
          \end{equation*}
    \item For $n = 6$, we have $r = 3$, $q = 3$ and
          \begin{equation*}
            \mathcal{L}_{6}(\rp) = \frac{1}{16}\partial_{z}\rp -\frac{1}{16} z {\partial_{z}}^{2}\rp+ \frac{1}{48} z^{2} {\partial_{z}}^{3}\rp =  \frac{1}{16} a_{1} - \frac{1}{16} a_{3}z^{2} +  \frac{5}{16} a_{5}z^{4}+ a_{6}z^{5}.
          \end{equation*}
  \end{itemize}
\end{ex}

In order to prove lemma~\ref{lem:Ln-even-operator}, we introduce the following notations. Given an infinite matrix $M = (M_{ij})_{i,j \ge 0}$, and two subsets $I : = \set{i_{1}, \dotsc , i_{p}}$, $J : = \set{j_{1}, \dotsc , j_{p}}$ of $\NN$, where $i_{1} < \dotsb < i_{p}$ and $j_{1} < \dotsb < j_{p}$, we define the square matrix $M(I,J)$ of size $p$
\begin{equation*}
  M(I,J)_{kl} = M_{i_{k}j_{l}}.
\end{equation*}
Besides, if $D$ is an infinite diagonal matrix $(\lambda_{0} \; \lambda_{1} \; \dotsb )$, we define $D(I)$ as the diagonal matrix $(\lambda_{i_{_{1}}} \, \lambda_{i_{2}} \dotsb \lambda_{i_{p}})$ of size $p$ .

The following result was proved in~\cite[Corollary 2]{GV1985}. If $A = (A_{ij})_{i,j \ge 0}$ is the binomial matrix, where
\begin{equation*}
  A_{ij} : = \begin{cases}
    \displaystyle{\binom{i}{j}}, \quad \text{if} \quad i \geq j, \\
    0, \quad \text{otherwise},
  \end{cases}
\end{equation*}
then $\det A(I,J) > 0$, provided that $j_{k} \le i_{k}$, for $1 \le k \le p$. Now, let $B = (B_{ij})_{i,j \ge 0}$ be the infinite matrix defined by
\begin{equation*}
  B_{ij} : = \begin{cases}
    \dfrac{i!}{(i-j)!}, \quad \text{if} \quad i \geq j, \\
    0, \quad \text{otherwise}.
  \end{cases}
\end{equation*}
and let $D$ be the infinite diagonal matrix $(0! \; 1! \; 2! \; \dotsb)$. Since
\begin{equation*}
  B(I,J) = A(I,J)D(J),
\end{equation*}
we get the following result.

\begin{lem}\label{lem:positive-minor-criteria}
  Let $I : = \set{i_{1}, \dotsc , i_{p}}$ and $J : = \set{j_{1}, \dotsc , j_{p}}$ be two subsets of $\NN$, with $i_{1} < \dotsb < i_{p}$ and $j_{1} < \dotsb < j_{p}$. If $j_{k} \le i_{k}$, for $1 \le k \le p$, then $\det B(I,J) > 0$.
\end{lem}

\begin{proof}[Proof of lemma~\ref{lem:Ln-even-operator}]
  We look for a solution $(\lambda_{1}, \dotsc , \lambda_{q})$ of~\eqref{eq:lambda-equations}. Set $X : = (\lambda_{1}, \dotsc, \lambda_{q})^{t}$,
  \begin{equation*}
    I^{1} : = \set{2, 4, \dotsc , 2(q-1), n}, \quad J : = \set{1,2, \dotsc , q} \quad \text{and} \quad B^{1} : = B(I^{1},J).
  \end{equation*}
  Then, \eqref{eq:lambda-equations} rewrites as
  \begin{equation*}
    B^{1} X = (0, \dotsc , 0, 1)^{t},
  \end{equation*}
  which has a unique solution $(\lambda^{(n)}_{1}, \dotsc, \lambda^{(n)}_{q})$ of rational numbers, because $B(I^{1},J)$ is an integer matrix, which is invertible by lemma~\ref{lem:positive-minor-criteria}. Now, if $\mathcal{L}_{n}$ is defined by~\eqref{eq:def-Ln} and
  \begin{equation*}
    \rp = a_{0}(x,y) + a_{1}(x,y)z + \dotsb + a_{n}(x,y)z^{n},
  \end{equation*}
  we have
  \begin{equation*}
    \mathcal{L}_{n}(\rp) = \sum_{i = 1}^{r} \alpha^{(n)}_{i}a_{2i-1}(x,y)z^{2i-2} + a_{n}z^{n-1},
  \end{equation*}
  where
  \begin{equation*}
    \alpha^{(n)}_{i} = \sum_{k = 1}^{2i-1} \dfrac{(2i-1)!}{(2i-1-k)!} \lambda^{(n)}_{k}, \qquad 1 \le i \le r.
  \end{equation*}
  Hence, if we set $Y : = (\alpha^{(n)}_{1}, \dotsc, \alpha^{(n)}_{q})^{t}$ (where $\alpha^{(n)}_{q} = 1$ if $n$ is odd),
  \begin{equation*}
    I^{2} : = \set{1, 3, \dotsc , 2q-1} \quad \text{and} \quad B^{2} : = B(I^{2},J),
  \end{equation*}
  we get
  \begin{equation*}
    Y = B^{2} X.
  \end{equation*}
  Therefore
  \begin{equation*}
    \alpha^{(n)}_{i} = \sum_{k = 1}^{q} B^{2}_{ik} \lambda^{(n)}_{k}
  \end{equation*}
  where
  \begin{equation*}
    \lambda^{(n)}_{k} = \frac{1}{\det B^{1}} (-1)^{k+q} \Delta^{1}_{qk},
  \end{equation*}
  and $\Delta^{1}_{qk}$ is the $(q,k)$ minor of $B^{1}$. We have thus
  \begin{equation*}
    \alpha^{(n)}_{i} = \frac{1}{\det B^{1}} \sum_{k = 1}^{q} (-1)^{k+q} B^{2}_{ik} \Delta^{1}_{qk} = \frac{\det C^{i}}{\det B^{1}},
  \end{equation*}
  where $C^{i}$ is the $q \times q$ matrix obtained from $B^{1}$ by substituting its last row
  \begin{equation*}
    (B^{1}_{q1} \; B^{1}_{q2} \; \dotsb \; B^{1}_{qq})
  \end{equation*}
  by the $i$-th row of $B^{2}$
  \begin{equation*}
    (B^{2}_{i1} \; B^{2}_{i2} \; \dotsb \; B^{2}_{iq}).
  \end{equation*}
  For $n = 1$, we have $\alpha^{(1)}_{1} = 1$ and for $n = 2$, we have $\alpha^{(2)}_{1} = 1/2$ (see example~\ref{ex:Ln}). Hence the theorem is true for $n = 1,2$. If $n \ge 3$, we have $q\ge 2$ and, if necessary, by a circular permutation of signature $q-i$, we can uprise this last row so that the new matrix corresponds to $B(K^{i},J)$, where
  \begin{equation*}
    K^{i} : = \set{2, 4, \dotsc , 2i-2, 2i-1, 2i , \dotsc , 2(q-1)}.
  \end{equation*}
  We get finally that $\det C^{i} = (-1)^{q-i} \det B(K^{i},J)$, which doe not vanish by lemma~\ref{lem:positive-minor-criteria}, because $q \ge 2$. This achieves the proof.
\end{proof}

We are now ready to state a necessary and sufficient condition for a homogeneous polynomial of degree $n \ge 1$ to be even in $z$. In the next theorem, we use the Lie-Poisson bracket
\begin{equation*}
  \poisson{f}{g}_{LP}(\xx) : = \det (\xx, \nabla f, \nabla g).
\end{equation*}

\begin{thm}\label{thm:z-even-CNS}
  Let $\rp \in \Pn{n}(\RR^{3})$ be a homogeneous polynomial of degree $n \ge 1$ and let $\mathcal{L}_{n}$ be the differential operator defined by~\eqref{eq:def-Ln} with $\lambda^{(n)}_{k} = (B^{1})^{-1}_{qk}$. We have the following results.
  \begin{enumerate}
    \item If $n$ is odd, then, $\rp$ is $z$-even iff $\mathcal{L}_{n}(\rp) = 0$.
    \item If $n$ is even, then, $\rp$ is $z$-even iff $\poisson{\mathcal{L}_{n}(\rp)}{z}_{LP} = 0$.
  \end{enumerate}
\end{thm}

For the proof of theorem~\ref{thm:z-even-CNS}, the following result will be useful.

\begin{lem}\label{lem:vanishing-LP-bracket}
  Let $a(x,y)$ be an odd homogeneous polynomial in $(x,y)$. Then
  \begin{equation*}
    \poisson{a(x,y)}{z}_{LP} = 0 \iff a = 0.
  \end{equation*}
\end{lem}

\begin{proof}
  If $a$ vanishes identically, then $\poisson{a(x,y)}{z}_{LP} = 0$. Conversely, observe that
  \begin{equation*}
    \poisson{a(x,y)}{z}_{LP} = x\partial_{y} a - y\partial_{x} a = \partial_{\theta}a,
  \end{equation*}
  using polar coordinates $(r,\theta)$. But, $a(x,y)$ writes as $r^{2p+1}f_{p}(\theta)$. Hence, if $\poisson{a(x,y)}{z}_{LP} = 0$, then, $a = Cr^{2p+1}$ for some constant $C$ and since $r = \sqrt{x^{2} + y^{2}}$, this leads to $a = 0$.
\end{proof}

We will also make use of the following two properties of the Lie-Poisson bracket, which result from iteration of the Leibniz rule:
\begin{enumerate}
  \item $\poisson{z^{k}}{z}_{LP} = 0$, for every $k \in \NN$;
  \item $\poisson{f z^{k}}{z}_{LP} = \poisson{f}{z}_{LP}z^{k}$, for every function $f$ and every $k \in \NN$.
\end{enumerate}

\begin{proof}[Proof of theorem~\ref{thm:z-even-CNS}]
  Let
  \begin{equation*}
    \rp = a_{0}(x,y) + a_{1}(x,y)z + \dotsb + a_{n}(x,y)z^{n},
  \end{equation*}
  be a homogeneous polynomial of degree $n \ge 1$.

  (1) Suppose first that $n = 2r+1$ is odd. If $\rp$ is $z$-even, then $a_{n} = a_{2r+1}$ and all the terms $a_{2i-1}$ vanish and thus $\mathcal{L}_{n}(\rp) = 0$. Conversely, if $\mathcal{L}_{n}(\rp) = 0$, then,
  \begin{equation*}
    \mathcal{L}_{n}(\rp) = \sum_{i = 1}^{r} \alpha^{(n)}_{i}a_{2i-1}(x,y)z^{2i-2} + a_{2r+1}z^{n-1} = 0,
  \end{equation*}
  where each $\alpha^{(n)}_{i} \ne 0$ by lemma~\ref{lem:Ln-even-operator}. Thus $a_{2i-1} = 0$ for $1 \le i \le r+1$ and $\rp$ is $z$-even.

  (2) Suppose now that $n = 2r$ is even. If $\rp$ is $z$-even, then all the terms $a_{2i-1}$ vanish. Thus $\mathcal{L}_{n}(\rp) = a_{n}z^{n-1}$ where $a_{n} = a_{2r}$ is a constant. Hence, $\poisson{\mathcal{L}_{n}(\rp)}{z}_{LP} = a_{n}\poisson{z^{n-1}}{z}_{LP} = 0$. Conversely, if $\poisson{\mathcal{L}_{n}(\rp)}{z}_{LP} = 0$, then, we have
  \begin{equation*}
    \poisson{\sum_{i = 1}^{r} \alpha^{(n)}_{i}a_{2i-1}(x,y)z^{2i-2} + a_{2r}z^{n-1}}{z}_{LP} = \sum_{i = 1}^{r} \alpha^{(n)}_{i} \poisson{a_{2i-1}(x,y)}{z}_{LP} z^{2i-2} = 0.
  \end{equation*}
  Hence, we have $\poisson{a_{2i-1}(x,y)}{z}_{LP} = 0$ for $1 \le i \le r$. But since $n$ is even, all the $a_{2i-1}(x,y)$ are odd and by lemma~\ref{lem:vanishing-LP-bracket}, they must vanish. This achieves the proof.
\end{proof}

We are now interested to formulate a necessary and sufficient condition for a homogeneous polynomial $\rp$ of degree $n \ge 1$ to be odd in $z$. A result similar to lemma~\ref{lem:Ln-even-operator} will be established first. To do so, we introduce the differential operator
\begin{equation*}
  \widetilde{\mathcal{D}}_{k} : = z^{k}{\partial_{z}}^{k}, \qquad k = 0,1,2, \dotsc .
\end{equation*}

\begin{lem}\label{lem:Kn-odd-operator}
  Let $n \ge 1$ and set $q = \lfloor \dfrac{n+1}{2}\rfloor$, $r = \lfloor \dfrac{n}{2}\rfloor$, so that $q = r$, if $n$ is even and $q = r+1$, if $n$ is odd. Then, there exists a unique $(r+1)$-tuple $(\mu^{(n)}_{0}, \dotsc, \mu^{(n)}_{r})$ of rational numbers, solution of the equations (in the unknowns $(\mu_{0}, \dotsc, \mu_{r})$)
  \begin{equation*}
    \sum_{k = 0}^{r} B_{2i-1,k} \mu_{k} = 0, \quad \text{for} \quad i = 1, \dotsc , r \quad \text{and} \quad \sum_{k = 0}^{r} B_{n,k}\mu_{k} = 1  .
  \end{equation*}
  Moreover, if we define the linear operator
  \begin{equation}\label{eq:def-Kn}
    \mathcal{K}_{n} : = \sum_{k = 0}^{r} \mu^{(n)}_{k} \widetilde{\mathcal{D}}_{k},
  \end{equation}
  then, for every homogeneous polynomial $\rp$ of degree $\le n$, we have
  \begin{equation*}
    \mathcal{K}_{n}(\rp) = \sum_{i = 0}^{q-1} \beta^{(n)}_{i} a_{2i}(x,y)z^{2i} + a_{n}(x,y)z^{n},
  \end{equation*}
  where
  \begin{equation*}
    \beta^{(n)}_{i} = \sum_{k = 0}^{r} B_{2i,k} \mu^{(n)}_{k}  \ne 0, \qquad 0 \le i \le q-1.
  \end{equation*}
\end{lem}

Since the proof is almost identical to the one of lemma~\ref{lem:Ln-even-operator}, we will not repeat it but just emphasize the changes. In the proof of lemma~\ref{lem:Ln-even-operator}, the $q \times q$ matrix $B^{1}$ should be replaced by the $(r+1) \times (r+1)$ matrix $\tilde{B}^{1} := B(\tilde{I}^{1}, \tilde{J})$, where
\begin{equation*}
  \tilde{I}^{1} : = \set{1, 3, \dotsc , 2r-1, n}, \quad \tilde{J} : = \set{0,1,2, \dotsc , r}.
\end{equation*}
Note that $\det \tilde{B}^{1} >0$ by lemma~\ref{lem:positive-minor-criteria} and hence that the $\mu^{(n)}_{k}$ are uniquely defined. Then, we introduce the $q \times (r+1)$ matrix $\tilde{B}^{2} := B(\tilde{I}^{2},\tilde{J})$ where
\begin{equation*}
  \tilde{I}^{2} : = \set{0, 2, \dotsc , 2(q-1)}.
\end{equation*}
We get thus
\begin{equation*}
  \beta^{(n)}_{i} = \frac{1}{\det \tilde{B}^{1}} \sum_{k = 0}^{r} (-1)^{k+r+1} {\tilde{B}^{2}}_{i,k} \tilde{\Delta}^{1}_{r+1,k} = \frac{\det \tilde{C}^{i}}{\det \tilde{B}^{1}}, \qquad 0 \le i \le q-1,
\end{equation*}
where $\tilde{\Delta}^{1}_{r+1,k}$ is the $(r+1,k)$ minor of $\tilde{B}^{1}$ and $\tilde{C}^{i}$ is the $(r+1) \times (r+1)$ matrix obtained from $\tilde{B}^{1}$ by substituting its last row by the $i$-th row of $\tilde{B}^{2}$. As in lemma~\ref{lem:Ln-even-operator}, we can show that $\det \tilde{C}^{i} \ne 0$. We get therefore the following result, which proof is similar to that of theorem~\ref{thm:z-even-CNS} and will be omitted

\begin{thm}\label{thm:z-odd-CNS}
  Let $\rp \in \Pn{n}(\RR^{3})$ be a homogeneous polynomial of degree $n \ge 1$, $r = \lfloor \dfrac{n}{2}\rfloor$ and $q = \lfloor \dfrac{n+1}{2}\rfloor$. Let $\mathcal{K}_{n}$ be the differential operator defined by~\eqref{eq:def-Kn} with $\mu^{(n)}_{k} = (\tilde{B}^{1})^{-1}_{1k}$. We have the following results.
  \begin{enumerate}
    \item If $n$ is even, then, $\rp$ is $z$-odd iff $\mathcal{K}_{n}(\rp) = 0$.
    \item If $n$ is odd, then, $\rp$ is $z$-odd iff $\poisson{\mathcal{K}_{n}(\rp)}{z}_{LP} = 0$.
  \end{enumerate}
\end{thm}

\begin{ex}\label{ex:Kn}
  \begin{itemize}
    \item For $n = 1$, we have $r = 0$, $q=1$ and
          \begin{equation*}
            \mathcal{K}_{1}(\rp) = \rp =  a_{0} + a_{1}z.
          \end{equation*}
    \item For $n = 2$, we have $r = q= 1$ and
          \begin{equation*}
            \mathcal{K}_{2}(\rp) = -p + z\partial_{z}\rp = -a_{0} + a_{2}z^{2}.
          \end{equation*}
    \item For $n = 3$, we have $r = 1$, $q=2$ and
          \begin{equation*}
            \mathcal{K}_{3}(\rp) = -\frac{1}{2} p + \frac{1}{2} z\partial_{z}\rp = -\frac{1}{2} a_{0} + \frac{1}{2}a_{2}z^{2}+a_{3}z^{3}.
          \end{equation*}
    \item For $n = 4$, we have $r =q = 2$ and
          \begin{equation*}
            \mathcal{K}_{4}(\rp) = \rp -  z\partial_{z} \rp + \frac{1}{3} z^{2}{\partial_{z}}^{2}\rp = a_{0} - \frac{1}{3} a_{2}z^{2} + a_{4}z^{4}.
          \end{equation*}
    \item For $n = 5$, we have $r =2$, $q = 3$ and
          \begin{equation*}
            \mathcal{K}_{5}(\rp) = \frac{3}{8} \rp - \frac{3}{8} z\partial_{z} \rp + \frac{1}{8} z^{2}{\partial_{z}}^{2}\rp = \frac{3}{8} a_{0} - \frac{1}{8} a_{2}z^{2} + \frac{3}{8} a_{4}z^{4}+ a_{5}z^{5}.
          \end{equation*}
    \item For $n = 6$, we have $r =q = 3$ and
          \begin{equation*}
            \mathcal{K}_{6}(\rp) = -\rp + z\partial_{z} \rp -\frac{2}{5} z^{2}{\partial_{z}}^{2}+\frac{1}{15} z^{3}{\partial_{z}}^{3}\rp = - a_{0} + \frac{1}{5} a_{2}z^{2} - \frac{1}{5}a_{4}z^{4}
            + a_{6}z^{6}.
          \end{equation*}
  \end{itemize}
\end{ex}


\end{document}